\documentclass[11pt]{article}
\usepackage{graphicx}
\usepackage[left = 3cm, right = 3cm, top = 3cm, bottom = 3cm]{geometry}
\usepackage{subcaption}
\usepackage{wrapfig}
\usepackage{amsmath}
\usepackage{amsfonts}
\usepackage{amsthm}
\usepackage{dsfont}
\usepackage{hyperref}
\usepackage{longtable}
\usepackage{booktabs}
\usepackage{mathtools}
\usepackage{tcolorbox}
\usepackage{enumerate}
\usepackage{comment}
\usepackage{amssymb}
\usepackage{authblk}
\usepackage{xspace}
\usepackage{authblk}
\usepackage{appendix}
\usepackage{tabularx}
\usepackage{algorithm}
\usepackage{algpseudocode}
\usepackage{multirow}
\usepackage{adjustbox}

\newtheorem{observation}{Observation}

\theoremstyle{definition}
\newtheorem*{def*}{Definition}

\theoremstyle{definition}
\newtheorem*{prf*}{Proof}
\theoremstyle{definition}
\newtheorem*{prfthm*}{Proof of Theorem}

\newcommand{\Anil}[1]{{\noindent \textcolor{blue} {\textbf{Anil}: {#1}}}}

\title{Cascading Impacts of the USA–China Trade War on Global Oilseed Supply Chain}

\date{}
\author{\small Diksha Gupta$^1$, Ritwick Mishra$^{1,2}$, Achla Marathe$^{1,3}$, Krista Danielle Yu$^{4}$, Anil Vullikanti$^{1,2}$\\
\vspace{5pt}
$^1$Biocomplexity Institute,\\
$^2$Department of Computer Science,\\
$^3$Department of Public Health Sciences,\\
University of Virginia, Charlottesville, Virginia, USA\\
$^4$Department of Economics, De La Salle University, Manila, Philippines}

\begin{document}

\maketitle

\begin{abstract}
\noindent Global supply chains are highly interconnected, making them vulnerable to cascading disruptions induced by trade policy shocks. Understanding how such disruptions propagate through production networks, and how mitigation mechanisms such as trade reallocation and production adjustment can alleviate their impacts—remains a central challenge. In this work, we develop a linear programming formulation of an Input-Output (IO) system that captures cascading supply-chain disruptions together with trade reallocation and production expansion. Our formulation yields a system-level equilibrium characterization that enables the joint analysis of disruption propagation and mitigation within a unified framework. 

We propose an efficient algorithm for computing approximate equilibrium solutions by minimizing total unmet demand in large IO systems. We apply our approach to tariff-induced disruptions in the global oilseeds supply chain arising from the U.S.-China trade war. Our results show that a localized 70\% disruption to flows from the U.S. oilseeds sector to China leads to a 3.27\% loss in global output, with China experiencing a disproportionate loss of 14.02\%. As a counterfactual mitigation strategy, allowing a 20\% reallocation from Brazil oilseed to China significantly reduces global output losses to 1.36\%, although pressure remains high on final-demand flows. We further investigate production expansion as an additional mitigation mechanism and show that it introduces tradeoffs between reducing global final-demand losses and protecting Brazil’s domestic flows. Domestic reallocation disproportionately shifts losses toward smaller economies, while globally sourced expansion redistributes losses more broadly across the network.

\end{abstract}

\section{Introduction}
\label{sec:intro}

In today’s globalized economy, supply chains span multiple countries, making them highly sensitive to shifts in trade policy. Even localized disruptions can propagate through production networks, altering trade flows and reshaping market equilibria worldwide. These dynamics are particularly evident in agricultural commodities such as soybeans, where China depends heavily on imports from a small set of producers—Brazil (40\%), the United States (28\%), and Argentina (12\%), amounting to approximately 105 million tons in 2024 ($\sim$ 85\% of domestic demand)~\cite{peng2025global}.

Trade policy shocks can rapidly reconfigure these dependencies. For example, the 2018 U.S.–China trade war led to a sharp substitution away from U.S. soybeans toward Brazilian exports~\cite{colussi2025us_china_soybean}. More recently, following a six-month halt in U.S. imports in 2025, China increased sourcing from Brazil and Argentina, aided by temporary export tax relaxations in Argentina, before partially restoring trade flows through a 2026–2028 purchase agreement with the United States~\cite{USDAFASOilseedsReportDec2025}.


These cross-border trade shocks not only affect supply dynamics in importing countries but also propagate through the global soybean trade, as limited short-run production and export capacity restrict the system’s ability to absorb sudden demand shifts. Following the 2018 United States-China trade war, China’s soybean imports reallocated sharply away from the United States toward Brazil. This reallocation coincides with a pronounced increase in Brazil’s domestic soybean prices~\cite{cepea_soybean_indicator}. Notably, this imbalance was accompanied by a surge of over $20\%$ in Brazil’s soybean oil imports at higher market prices in 2021~\cite{colussi2025us_china_soybean}. Motivated by this trade-induced imbalance in resource allocation, the current study focuses on modeling the cascading impacts of trade disruptions, with particular attention to how shocks reallocate resources between competing end uses and induce production adjustments.



Recent works have studied the impact of trade disruptions as a function of the pricing-quantity equilibria using Computable General Equilibrium models ~\cite{eugster2022effect,grossman2024tariffs,auclert2025macroeconomics}.While these studies capture economy-wide interactions through price and quantity adjustments, it is difficult to represent physical bottlenecks, capacity constraints, and cascading shortages that arise in the interconnected global networks. Elobeid et al. ~\cite{elobeid2021jae} analyzed the impacts of China's retaliatory tariffs using a partial equilibrium model combined with input-output analysis.

To explicitly represent inter-sectoral production dependencies and capacity-driven propagation of shocks, an alternative modeling approach treats the global supply chain as a network of input–output relationships. The representation of the global supply chain as an \textit{input-output (IO) system} originates from Leontief’s formulation of interdependent production economies \cite{leontief1986input}. This framework became the canonical approach for representing global production systems, encoding intermediate production dependencies across countries and sectors within a unified accounting structure. Today, IO-based models underpin most large-scale analyses of global supply chains, trade shocks, and cascading disruptions.

Many existing studies analyze supply chain disruptions using agent-based modeling (ABM) built on IO representations~\cite{guan2020global,mu2021robustness,cao2025data,gomez2021supply,moran2025critical,yabe2025behaviour}. These approaches explicitly simulate the propagation of disruptions over time through interactions among heterogeneous economic agents. While ABMs are well suited for capturing localized dynamics and scenario-specific adaptation mechanisms, they rely heavily on calibration choices and behavioral assumptions. As a result, they do not naturally yield global equilibrium characterizations and scale poorly to economy-wide settings, limiting their applicability for analyzing system-level effects of large-scale disruptions.

\smallskip

To address these limitations, we formulate the input-output (IO) model as a \textit{linear program (LP)} to analyze the cascading impact of trade  disruptions on the global supply chain. Existing works have widely used LP based techniques for supply-chain optimizations~\cite{vogstad2009input}. Santos et al. use LP to study the system-wide impact of disruptions to interconnected economies in~\cite{santos2006inoperability}. Recent work by Soltanisehat et al. ~\cite{soltanisehat2024multiregional} proposes a novel multiobjective mixed-integer linear programming formulation for optimizing lockdowns policy for the US economy during COVID-19.

In contrast to prior LP-based IO studies that primarily analyze static disruption impacts, we model trade shocks together with endogenous trade reallocation and production adjustment, enabling the study of both cascading effects and mitigation pathways within a unified framework. This formulation enables us to quantify how policy-driven adjustments propagate through the supply chain and provides insights into the design of more resilient and adaptive trade strategies. Our analysis is explicitly counterfactual: we recompute global production–trade equilibria under alternative trade and mitigation scenarios to isolate the system-level effects of policy-induced disruptions.

The China-US trade dispute has been analyzed using price-based partial or general equilibrium models~\cite{elobeid2021jae}. These models explain trade diversion through price adjustment and substitution, with prices adjusting to ensure that markets clear. 
Recent work has extended multi-regional input-output systems using optimization techniques that include nonlinear and linear programming. However, these models typically preserve feasibility by allowing smooth substitution across trade partners~\cite{bonfiglio2025impact}, focus on identifying trade that minimize specific outcomes such as emissions~\cite{yan2025buy}, or treat trade as a buffer against exogenous production disruptions~\cite{He2019energyecon}.

\medskip
\noindent\textbf{Contributions.}
This work makes three main contributions. 
First, we propose a linear programming formulation of an Input-Output system that captures cascading supply-chain disruptions together with \emph{endogenous} trade reallocation and production adjustment under capacity constraints. Unlike prior IO-based LP models that primarily analyze static disruption impacts, our formulation enables the study of both disruption propagation and mitigation pathways within a unified equilibrium framework. Second, we apply our framework to a large-scale global case study of tariff-induced disruptions in the oilseeds supply chain arising due to the US-China trade wars, using GTAP MRIO data, and conduct counterfactual analyses that quantify how trade reallocation and production expansion differentially mitigate cascading losses across regions and end uses. Third, using our approach we demonstrate that trade diversion is conditional, that is, when trade and capacity constraints limit adjustment, reallocation may be insufficient to fully offset the trade disruption which will result to unmet final demand. This shifts the analysis from price-mediated adjustment to the structural limits of global supply chains which are not captured by price-clearing models.

\section{Methodology}\label{sec:method}

We begin by introducing the notation used in this work in Table~\ref{tab:notation}. We denote sets by calligraphic letters; parameters and constants are denoted by capital letters; and optimization variables are denoted by small letters. 
\begin{table}[!ht]
\centering
\caption{Summary of notation.}
\begin{tabular}{p{0.18\linewidth} p{0.75\linewidth}}
\toprule
\textbf{Symbol} & \textbf{Description} \\
\midrule

\textbf{Parameters} & \\ [3pt]

$\mathcal{F}$ & Set of firms. \\

$\mathcal{S}$ & Set of sectors (each firm belongs to exactly one sector). \\

$\mathcal{H}$ & Set of households (one household per country). \\

$\sigma: \mathcal{F} \to \mathcal{S}$ & Mapping assigning each firm to its output sector. \\

$P_i$ & Primary inputs available to firm $i$ (e.g., labor, resources, capital). \\

$A_{s,j}$ & Leontief input coefficient: input from sector $s$ used by firm $j$. \\

$B_i$ & Leontief coefficient for primary inputs required by firm $i$. \\

$H_{i,k}^*$ & Household demand from household $k \in \mathcal{H}$ for the output of firm $i$ in base model. \\

$Z_{i,j}^*$ & Intermediate flow from firm $i$ to firm $j$ in base model. \\

$X_i^*$ & Output of firm $i$ in base model. \\

$\Delta_{i,j}, \Delta_{i,k} \in [0,1]$ & Fraction of disruption in the flow from firm $i$ to firm $j$ \& household $k$ induced by tariff policy ($1 =$ full disruption, $0 =$ none). \\

$\alpha_{i,j}, \alpha_{i,k} \ge 0$ & Elasticity factor for substitution of firm $i$'s output to firm $j$ \& household $k$. \\

$\beta_i \ge 0$ & Elasticity factor for increasing primary input use of firm $i$. \\

$\eta_i \ge 0$ & Elasticity factor for increasing production of firm $i$. \\[10pt]

\textbf{Variables} & \\[3pt]

$x_i \ge 0$ & Output of firm $i$. \\

$z_{i,j} \ge 0$ & Intermediate flow from firm $i$ to firm $j$. \\

$h_{i,k} \ge 0$ & Household supply from firm $i$ to household $k$. \\

$G \ge 0$ & Gap between total household supply and total household demand introduced due to disruption. \\

\bottomrule
\end{tabular}
\label{tab:notation}
\end{table}

\subsection{System Description} 

We model the global supply chain as a network of countries, sectors, firms, and households, with the following structure:
\begin{itemize}
    \item Each firm corresponds to a \emph{single sector in a single country}. Thus, a firm is uniquely identified by a country-sector pair. While a country may have many sectors, and a sector may exist in multiple countries, each sector within each country is represented by exactly one firm in our model.

    \item Firms produce output using (i) intermediate inputs sourced from other firms across the globe and (ii) primary inputs to which they have direct access, such as labor, capital, and natural resources.

    \item Each firm must produce sufficient output to satisfy both the intermediate-input demand of downstream firms and the final household demand.

    \item Firm production is governed by a Leontief production function, and inter-firm flows are governed by a standard Input–Output (IO) model.

    \item We model each sector in each country by a single representative firm, and a single representative household for each country.
\end{itemize} 

\begin{figure}[t]
    \centering
    \includegraphics[width=0.75\linewidth, clip, trim = 2cm 10cm 2cm 4cm]{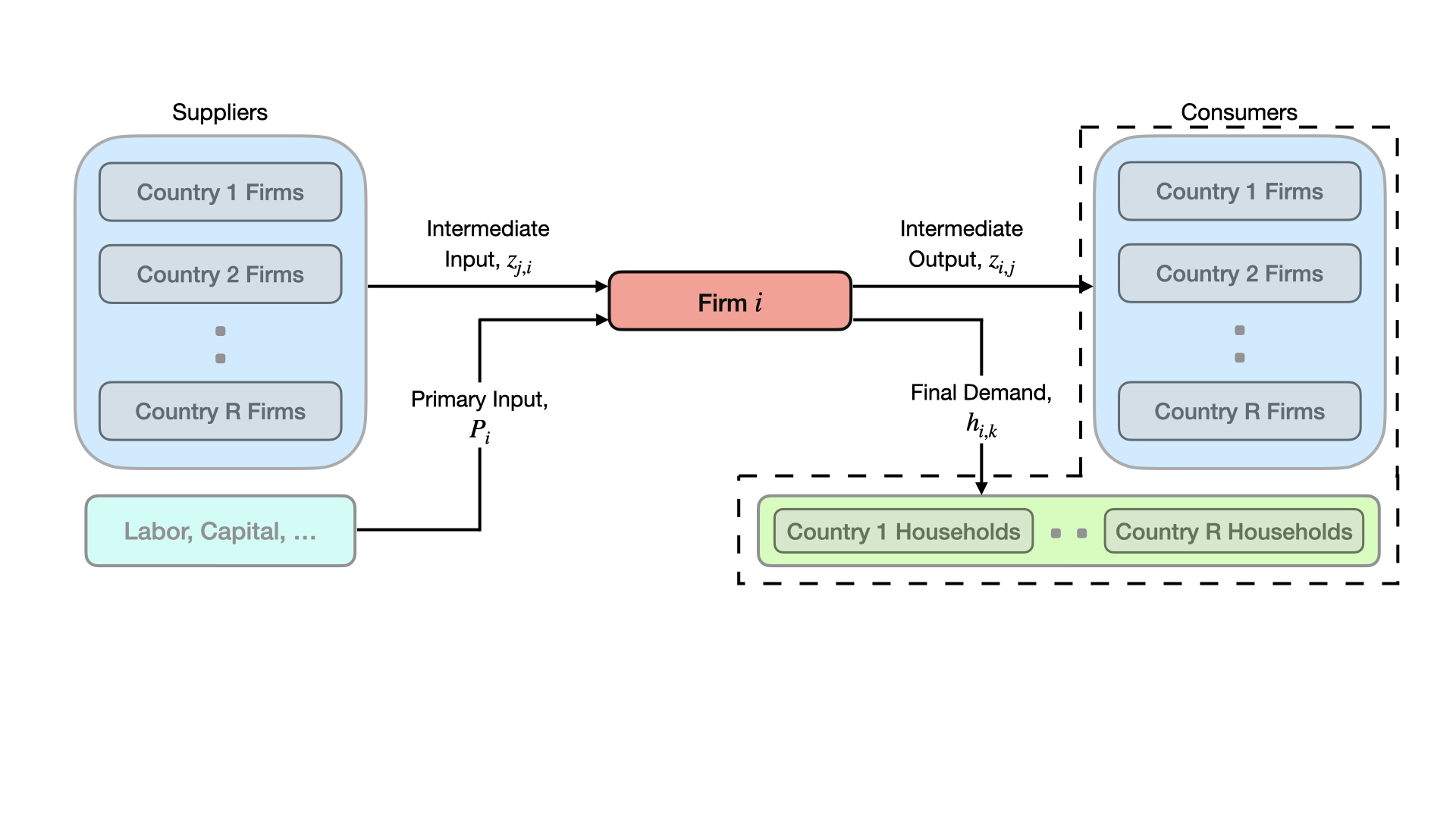}
    \caption{Schematic representation of the input--output structure for a representative firm. Firm $i$, corresponding to a country--sector pair, produces output $x_i$ using intermediate inputs $z_{j,i}$ from supplier firms and primary inputs $P_i$. Its output is allocated to consumer firms as intermediate supply $z_{i,j}$ and to households as final demand $h_{i,k}$.} 
    \label{fig:system_description}
\end{figure}

We further simplify the system by modeling only the flows between firms, treating country identifiers as metadata that are not directly relevant to the algorithmic design. Figure~\ref{fig:system_description} illustrates the resulting system abstraction. Country and sector distinctions are later reinstated at the firm level for the empirical case study in Section~\ref{sec:empirical}.

\subsection{Modeling trade disruptions and mitigation strategies in IO systems}
\label{sec:disruption-LP}

Prior work on estimating the impact of disruptions in input–output (IO) systems, such as \cite{guan2020global}, typically relies on iterative update mechanisms to compute post-disruption outcomes. These approaches are well suited for analyzing disruption scenarios in which production levels are reduced exogenously and no endogenous reallocation or production expansion is permitted. In such settings, the propagation of shocks is driven purely by fixed inter-sectoral dependencies.

However, modeling realistic responses to large-scale trade disruptions requires incorporating mitigation mechanisms, including the reallocation of trade flows across alternative suppliers and the expansion of production subject to capacity constraints. Introducing these mechanisms poses several challenges. First, re-allocations must respect the underlying flow conservation and interdependency structure of the production network to avoid infeasible solutions. Second, production expansion is constrained by sector- and country-specific capacity limits on the primary inputs as well as received intermediate inputs, making it insufficient to simply redistribute unmet demand. Finally, these adjustments interact globally: changes in one part of the network can induce secondary bottlenecks and cascades elsewhere, necessitating a system-level equilibrium formulation.

The linear programming based formulation is illustrated in Figure~\ref{fig:global_model-algo-increase} addresses these challenges. We describe the role and interpretation of each of the constraints from our approach below.

\begin{itemize}
    \item \textit{Leontief Constraints.}  The Leontief constraints, together with the objective of maximizing aggregate output, encode the Leontief production function.The coefficient $A_{s,i}$ specifies the amount of sector-$s$ intermediate input required per unit output of firm $i$, and the constraint ensures that total intermediate input sourced from all firms $j$ with $\sigma(j)=s$ is sufficient to support production level $x_i$.  
    Similarly, the primary input constraint $B_i x_i \le (1+\beta_i) P_i$ limits output by the availability of primary inputs. The parameter $\beta_i \ge 0$ enables controlled production expansion by allowing proportional increases in available primary inputs relative to the baseline.

    \item \textit{IO Constraints.}  
    These constraints link production levels to outgoing flows. The constraint enforces the classical IO balance condition, ensuring that total intermediate outflows along with the total household flow does not exceed the total output for any firm in the system.

    \item \textit{Disruption Constraints.} 
    These constraints model exogenous shocks by limiting intermediate flows and household flows to a fraction of their baseline values $Z^*_{i,j}$ and $H^*_{i,k}$, respectively. The parameters $\Delta_{i,j}, \Delta_{i,k} \in [0,1]$ specify the severity of the disruption on the flow from firm $i$ to $j$ for the former and from firm $i$ to household $k$ in the latter case. These parameters force a proportional reduction relative to the pre-disruption flow. 

    \item \textit{Reallocation Constraints.}  
    These constraints restrict the extent to which firms not directly disrupted can adjust their intermediate and household flows. The parameters $\alpha_{i,j}$ and $\alpha_{i,k} \in [0,1]$ controls the allowable deviation from baseline flows $Z^*_{i,j}$ and $H^*_{i,k}$, respectively, ensuring that post-disruption reallocations remain consistent with pre-existing economic relationships and capacity limits.

    \item \textit{Production Constraints.} We model production expansion as a mitigation mechanism by imposing firm-specific bound on output within the feasible region of $\mathcal{P}(G)$. These constraints represent exogenous capacity expansions at selected firms. In conjunction with reallocation constraints, the resulting LP characterizes how policy-driven production increases propagate through the global supply chain via adjustments in intermediate and household supply flows.

    \item \textit{$\beta$-Adjustment Constraint.} To enable production expansion, we require primary input to satisfy the Leontief constraint. We model this by increasing the available primary input by $\beta_i$, which is then determined as a function of production expansion parameter $\eta_i$. 
    
    \begin{figure}[t!]
    \centering
    \begin{tcolorbox}[colback=gray!10, colframe=gray!50, boxrule=0.8pt, left=4pt, right=4pt, top=6pt, bottom=6pt]
    
    For $G\geq 0$, the polytope, denoted by $\mathcal{P}(G)$, is defined by the following constraints.
    For all firms $i \in \mathcal{F}$
    \begin{flalign}
        &\textit{Leontief constraints:} 
        && A_{s,i}\, x_i 
            \;\le\; 
            \sum\nolimits_{\substack{j \in \mathcal{F}: \sigma(j)=s}} z_{j,i}, \qquad \forall s \in \mathcal{S}; &\nonumber \\
        &
            && B_i x_i \le (1+\beta_i) P_i^*,
            &\\[4pt] 
        &\textit{IO constraints:} 
        && \sum\nolimits_{j \in \mathcal{F}} z_{i,j} + \sum\nolimits_{k\in \mathcal{H}}h_{i,k} 
            \;=\; x_i,
            &&\nonumber \\[4pt]
    &\textit{Disruption constraints:} 
        && z_{i,j} 
            \;\le\; (1-\Delta_{i,j})\;\, Z^*_{i,j},
            \qquad \forall j \in \mathcal{F} \text{ with } \Delta_{i,j} > 0, \nonumber\\
    & && h_{i,k} \;\leq\, (1-\Delta_{i,k})\, H_{i,k}^*, \qquad \forall k \in \mathcal{H} \text{ with } \Delta_{i,k} > 0,
            &&\nonumber \\[4pt]
    &\textit{Reallocation constraints:} 
        && z_{i,j} 
            \;\le\; (1+\alpha_{i,j}) \;\, Z^*_{i,j},
            \qquad \forall j \in \mathcal{F} \text{ with } \Delta_{i,j} = 0,
            &&\nonumber \\
    &
        && h_{i,k} 
            \;\le\; (1+\alpha_{i,k}) H^*_{i,k},
            \qquad  \forall k \in \mathcal{H} \text{ with } \Delta_{i,j} = 0,
            &&\nonumber\\
    &\textit{Production constraint:} && x_i \;=\; (1 + \eta_i)\; X_i^*, \qquad \qquad  \text{if}\; \eta_i > 0 && \nonumber \\
    &\textit{$\beta$-Adjustment constraint:} && \beta_i \;=\; \textstyle (1 + \eta_i)\frac{B_iX_i^*}{P_i^*}-1,\qquad\text{\hspace{-4pt}if}\; \eta_i > 0 && \nonumber \\
    &\text{All households }k \in \mathcal{H}:&\nonumber\\
    &\textit{Demand constraints:} 
        && \sum\nolimits_{j \in \mathcal{F}: \sigma(j)=s} h_{j,k} \;\leq\; \sum\nolimits_{j \in \mathcal{F}: \sigma(j)=s} H_{j,k}^*, \qquad \forall s \in \mathcal{S};
            &&\nonumber \\[4pt]
    &\textit{Slack constraint:} 
        && \sum\nolimits_{j \in \mathcal{F}, k \in \mathcal{H}} h_{j,k} + G \;=\; \sum\nolimits_{j \in \mathcal{F}, k \in \mathcal{H}} H_{j,k}^*, 
            && \nonumber
    \end{flalign}

    Let $X = \sum_i x_i$ denote the global supply chain output of the system. Then,
    
    \begin{itemize}
    \item Phase 1: Output $X^\star := \max \; \{ X:(x,z,h,G)\in \mathcal P(G),\; G\ge 0 \}$.
    \item Phase 2: Output $G^\star := \min \; \{ G: (x,z,h,G)\in \mathcal P(G),\; G\ge 0,\; X \ge X^\star\} $.
    \end{itemize}
    
    \end{tcolorbox}
    \caption{Linear programming formulation for modeling cascading supply-chain disruptions and mitigation through trade reallocation and production expansion. The first phase maximizes feasible global output, while the second phase minimizes unmet final demand among output-maximizing allocations.}
    \label{fig:global_model-algo-increase}
\end{figure}

    \item \textit{Demand Constraint.}  
    These constraints limit the sector-wise flow to all household $k$ by the MRIO values. This enables the demand to not be met exactly for all households, and treats the household flows as variables. 

    \item \textit{Slack Constraint.} We introduce a system level slack variable, $G$ that captures the total unmet global household demand in the system. By allowing the global household demand to be partially unmet, the model departs from standard IO models where final demand is treated as a fixed parameter. Instead, demand is treated as a decision variable constrained by the MRIO values, which makes the LP flexible enough to represent supply shortages due to disruptions, and mitigation strategies through limited reallocation capacity, substitution and increased production.
\end{itemize}

The optimization LP formulation enables the system to endogenously identify which sectors and countries experience demand gaps under disruption, and to quantify the impact of targeted reallocation as well as increased production strategies in preserving the largest possible level of global output while simultaneously minimizing total unmet demand. In the current work, motivated by the importance of prioritizing the total economic activity in the global supply chain. Figure~\ref{fig:global_model-algo-increase} illustrates our two phase Linear program for maximizing the global system output and the resulting minimum feasible household disruption. 

\smallskip

With respect to our LP model, it is important to note a limitation in how we model household demands. The LP does not encode any notion of differential “value of life” or distributional priority when reducing unmet household demand: all units of unmet final demand are treated uniformly across countries and households. As a result, the optimal allocation is driven primarily by the objective of maximizing global output (and only secondarily by slack reduction), rather than by welfare weights that might prioritize certain populations or essential needs in practice. In the real world, policy responses may implicitly or explicitly value unmet demand differently across countries (e.g., due to income, vulnerability, critical goods, or humanitarian considerations), whereas our formulation enforces no such heterogeneity. Consequently, the household-loss outcomes should be interpreted as efficiency-driven, rather than as recommendations for how societies ought to prioritize demand satisfaction.

\clearpage
\begin{observation}
\label{obs:disruption-LP-BO}
    Any solution $ (x^*, z^*, h^*, G^*)$ computed by the LP from Figure~\ref{fig:global_model-algo-increase} satisfies the Leontief IO model.
\end{observation}
\begin{proof}
    At the end of Phase 1, the optimal solution satisfies the Leontief constraints for each firm in conjunction with maximizing the total global output as the objective, implying
    $$\textstyle x_i^* \leq \min\left\{ \min_{s \in \mathcal{S}} \Big\{\frac{1}{A_{s,i}}\sum_{ j:\sigma(j)=s} z^*_{j,i}\Big\} \; , \;  \frac{\beta_i P_i}{B_i}\right\}.$$
    This corresponds to the LP solution being a feasible solution for the Leontief IO model. 

    Next, in Phase 2 note that the optimal solution minimizes the total slack in the system while satisfying the constraint that global output, $X\geq X^*$, thus enforcing the max condition from Phase 1. Hence, the Leontief IO model is not violated.
\end{proof}

\subsection{Scaling the method} 

Despite the two phase search for the optimal global output and minimal slack, the feasibility check of the linear program is computationally intensive for large IO tables. To address this challenge, we apply a value-based sparsification of the intermediate flow matrix, removing flows below a threshold. This trade flow value based thresholding approach has been widely used in literature to reduce the dimensionality of the IO table while successfully conserving overall structure~\cite{simburger2023filter}. Formally, let $\tau$ be a flow threshold such that, for all $i,j \in \mathcal{F}$,
    \[ \textstyle 
        Z_{i,j}^\tau \;=\;
        \begin{cases}
            0, & \text{if } Z_{i,j}^* < \tau,\\[2pt]
            Z_{i,j}^*, & \text{otherwise.}
        \end{cases}
    \]
Using these updated intermediate flows and leveraging the separability of the Leontief production function across sectors, we compute the equilibrium outputs, $X_i^\tau$, revised  intermediate input coefficients for each sector $s$, $A_{s,i}^\tau$ and primary input coefficients, $B_i^\tau$ as:
\begin{align}
    \textstyle X_i^\tau &= \textstyle \sum_{j\in \mathcal{R} \times \mathcal{S}} Z_{i,j}^\tau + \sum_{j \in \mathcal{R}} H_{i,j}\, \qquad A_{s,i}^\tau \; &= \textstyle \; \frac{\sum_{j \in \mathcal{R} \times \{s\}}Z_{j,i}^\tau}{X_i^\tau}, \quad \text{and} \quad  B_{i}^\tau \; = \; \frac{P_i}{X_i^\tau}.
\end{align}

\noindent To study the impact of our sparsification strategy on the base MRIO, we define the value-added loss for the intermediate flow of firm $i \in \mathcal{F}$ as
    \[ 
        \text{Intermediate Value-Added Loss}_i(\tau)
        \;=\;
        \left(\frac{\sum_{j\in \mathcal{F}}Z_{i,j}^* - \sum_{j\in \mathcal{F}}Z_{i,j}^\tau}{\sum_{j\in \mathcal{F}}Z_{i,j}^*}\right) \times 100,
    \]
where $Z_{i,j}^\tau$ is the intermediate flows of firm $i$ under threshold $\tau$ and $Z_{i,j}^*$ is the original MRIO intermediate flows.

\section{Empirical Setup}\label{sec:empirical}

To illustrate the cascading impact of USA-China trade disruption, we focus on a bilateral trade flow of strategic importance: oilseed exports from US to China. China is the largest importer of US soybeans and related oilseed products. Disruptions to this flow propagate across multiple sectors in both economies, with effects that amplify throughout the global supply chain network. It provides a setting where we can study how disruptions cascade through interconnected production networks, how firms and regions reallocate production in response, and how our formulation captures these systemic effects. Using this case study, we apply our methodology to assess the consequences of disrupting U.S. oilseed exports to China and to quantify the extent to which resulting losses can be mitigated through feasible reallocation and increased production.

To enable substitution as a mitigation mechanism under disruption, we make a simplifying assumption that flows between different regions within the same sector can be seamlessly redirected. In practice, the oilseeds sector encompasses heterogeneous crops whose production is often region-specific, which can limit substitution across regions. To strike a balance between realism and tractability, we restrict attention to the major soybean-exporting countries i.e. the US and Brazil, and allow substitution only among these producers in this study.

We consider four disruption and mitigation scenarios based on a supply shock to United State's oilseed exports to China. In all cases, we induce a disruption on the USA-OSD $\rightarrow$ CHN trade flow with disruption scale $\Delta$. 

\emph{Scenario I} serves as the baseline disruption case with no mitigation, in which no trade reallocation or production expansion is allowed. \emph{Scenario II} allows for  reallocation from Brazil's oilseed exports to China (BRA-OSD $\rightarrow$ CHN) with substitution (or reallocation) scale $\alpha \in \{0.20,0.25,0.30\}$. \emph{Scenario III} studies production expansion by jointly relaxing primary-input constraints and imposing an explicit output increase of $\eta$ for Brazil’s oilseed sector, while allowing reallocation of domestic supplies  across Brazil to its oilseed sector. Finally, \emph{Scenario IV} expands the production through global increase in flow capacity to Brazil oilseed sector.  

Together, these scenarios offer a systematic evaluation of cascading disruption effects and the relative effectiveness of trade reallocation and production expansion as mitigation mechanisms. We demonstrate our methodology on the Global Multi-Region Input–Output (MRIO) Table built from the GTAP (Global Trade Analysis Project) Version 12 for the year 2023 \cite{aguiar2022global}.

\begin{table}[t!]
\vspace{-20pt}
\centering
\caption{Experimental Scenarios and Parameters.}
\label{tab:scenarios}
\resizebox{\textwidth}{!}{
\begin{tabular}{lcccc}
\hline
\textbf{Parameters} & \textbf{Scenario I} & \textbf{Scenario II} & \textbf{Scenario III}& \textbf{Scenario IV}\\ 
\hline
Disrupted Country 
& USA--OSD $\rightarrow$ CHN 
& USA--OSD $\rightarrow$ CHN 
& USA--OSD $\rightarrow$ CHN 
& USA--OSD $\rightarrow$ CHN \\

Disruption Scale, $\Delta$ 
& 0.50, 0.60, \textbf{0.70} 
& 0.50, 0.60, \textbf{0.70} 
& \textbf{0.70} 
& \textbf{0.70}\\

Reallocation 
& -- 
& BRA--OSD $\rightarrow$ CHN 
& BRA--OSD $\rightarrow$ CHN 
& BRA--OSD $\rightarrow$ CHN \\

&&& BRA--* $\rightarrow$ BRA-OSD & * $\rightarrow$ BRA-OSD \\

Reallocation Scale, $\alpha$ 
& --  
& \textbf{0.20},0.25,0.30 
& \textbf{0.20}
& \textbf{0.20} \\

Increased Production 
& -- 
& -- 
& BRA--OSD
& BRA--OSD \\

Production Scale, $\eta$ 
& -- 
& -- 
&\textbf{0.10}
&\textbf{0.10} \\

\hline
\end{tabular}
}
We include a discussion of the experimental results for the values highlighted in bold.\\ Results for remaining values are provided in the Appendix.
\end{table}

\begin{figure}[t!]
    \centering

    \begin{subfigure}{0.48\linewidth}
        \centering
        \includegraphics[width=0.49\linewidth]{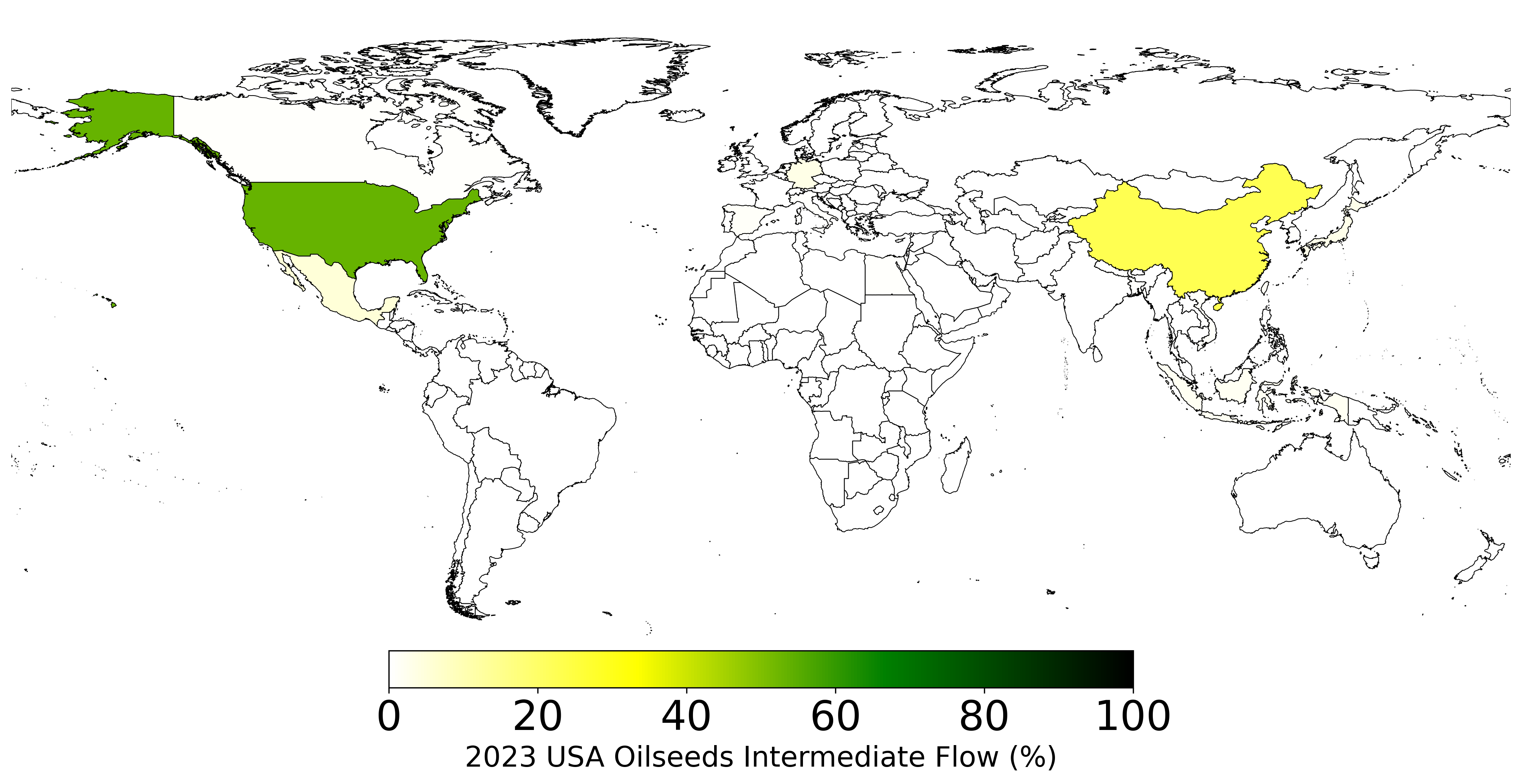}
        \includegraphics[width=0.49\linewidth]{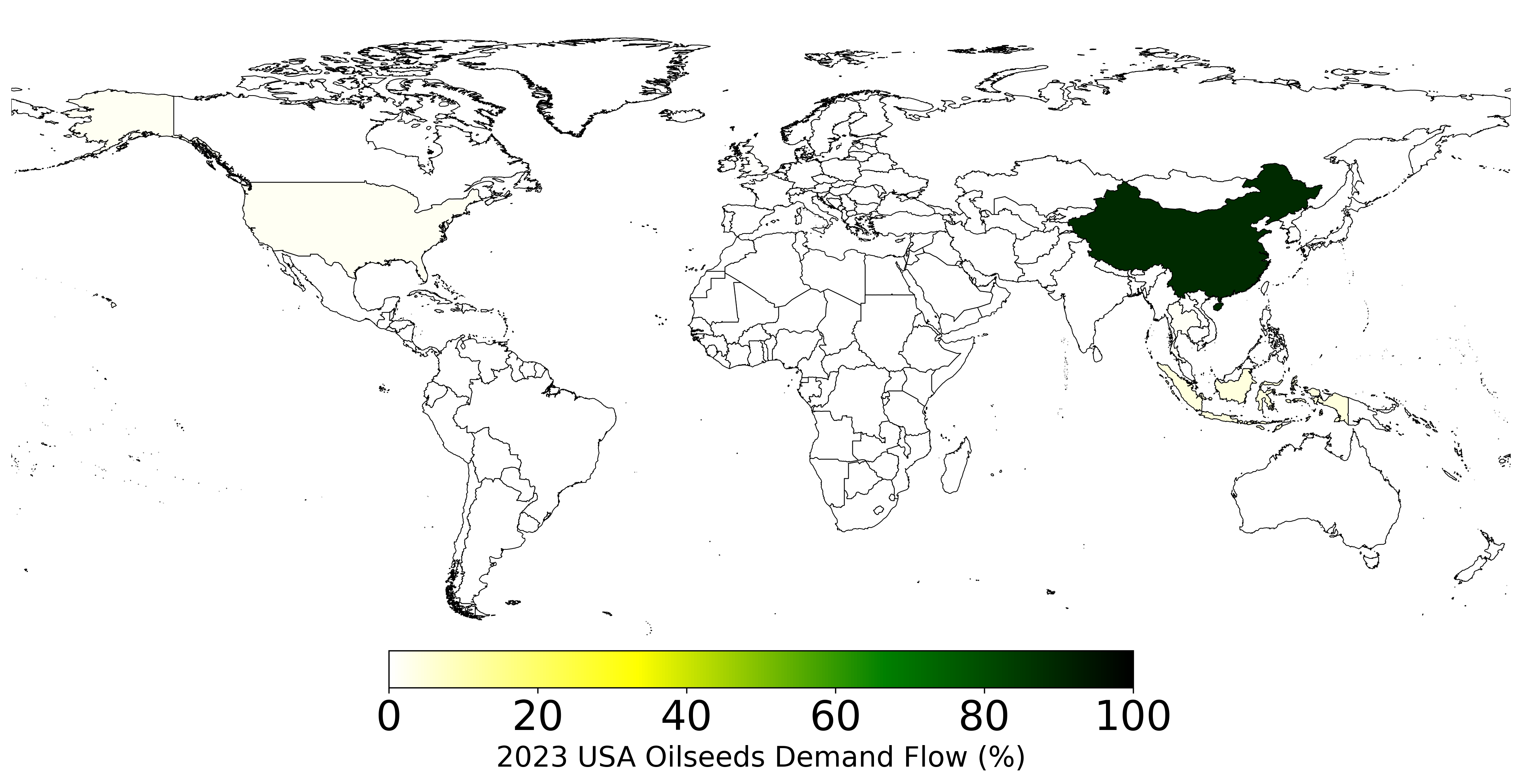}
        \caption{USA Oilseed flows to Consumers.}
    \end{subfigure}
    \hfill
    \begin{subfigure}{0.48\linewidth}
        \centering
        \includegraphics[width=0.49\linewidth]{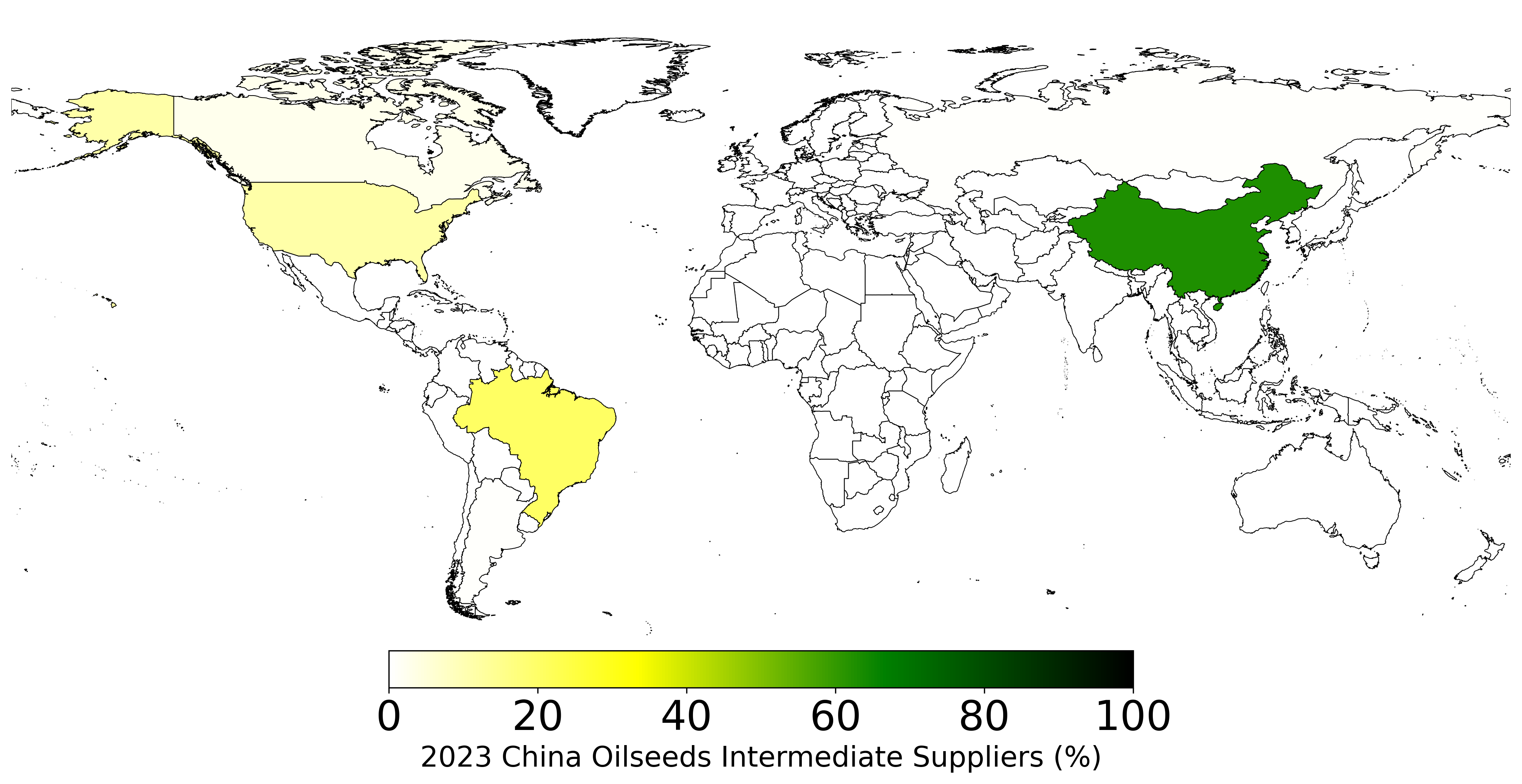}
        \includegraphics[width=0.49\linewidth]{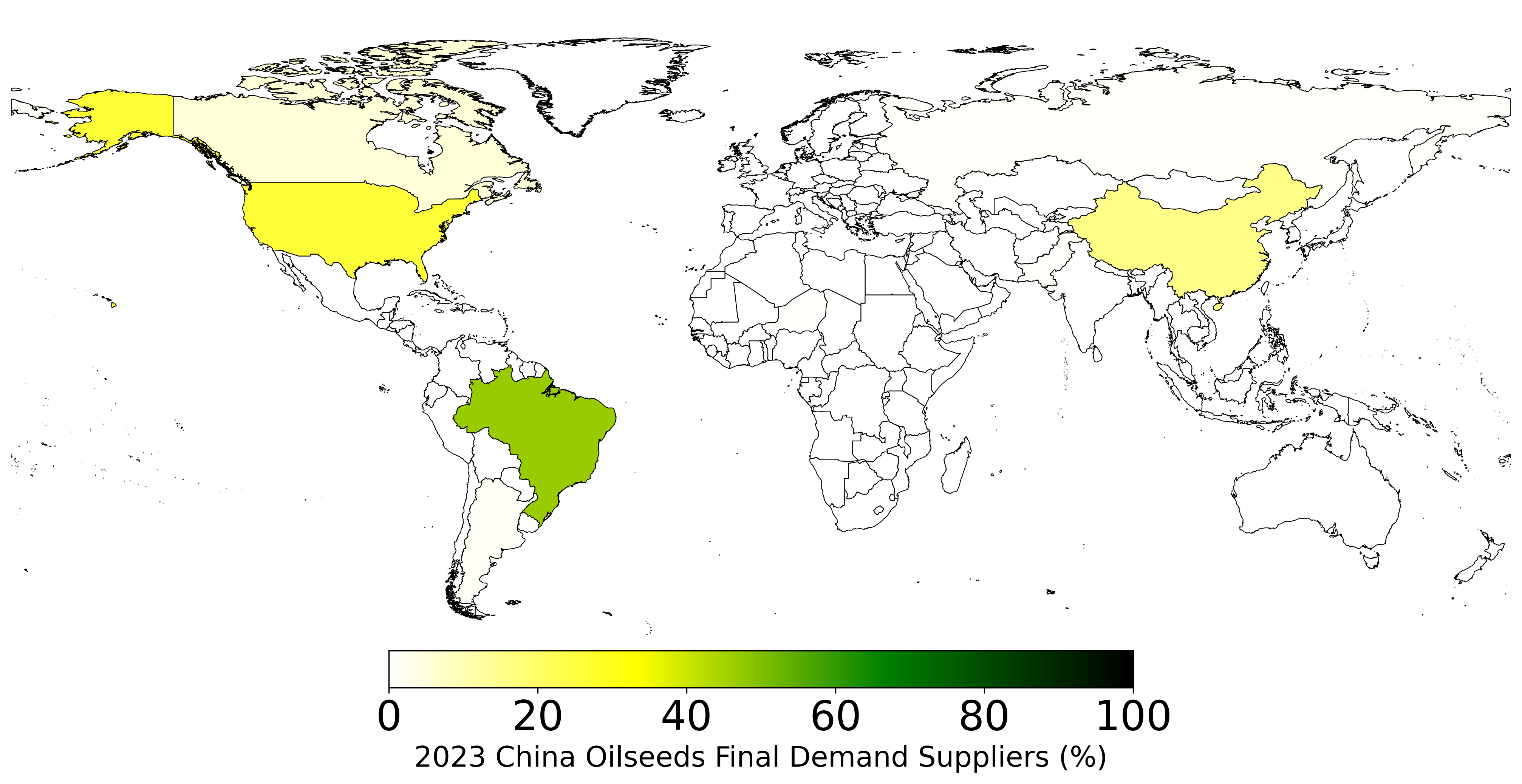}
        \caption{China Oilseed flows from Suppliers.}
    \end{subfigure}

    \caption{Baseline MRIO USA \& China Oilseed trade distribution in the Year 2023.}
    \label{fig:USA_OSD_flows}
\end{figure}

\section{Results}

We measure the economic impact in terms of the value-added loss for the entity with respect to the sparsified MRIO, unless mentioned otherwise. \emph{We measure economic activity using gross output derived from MRIO tables. Note that gross output differs from GDP, as it includes intermediate production flows across sectors. GDP corresponds to value-added and excludes intermediate consumption, whereas our measure captures total production activity within the network.}

Unless otherwise specified, we present results for four scenarios and report additional cases in the Appendix. We fix Scenario I to model a $70\%$ disruption in U.S. oilseed exports to China. Scenario II mitigates this disruption through a $20\%$ reallocation of Brazilian oilseed supply to China. Scenario III further supplements this reallocation with a $10\%$ increase in Brazilian oilseed production supported by domestic inputs, while Scenario IV instead allows this production expansion to be supported by global inputs.

\begin{figure}[b!]
\centering
\scriptsize

\includegraphics[width = 0.9\linewidth, clip, trim = 18cm 14cm 17cm 11cm]{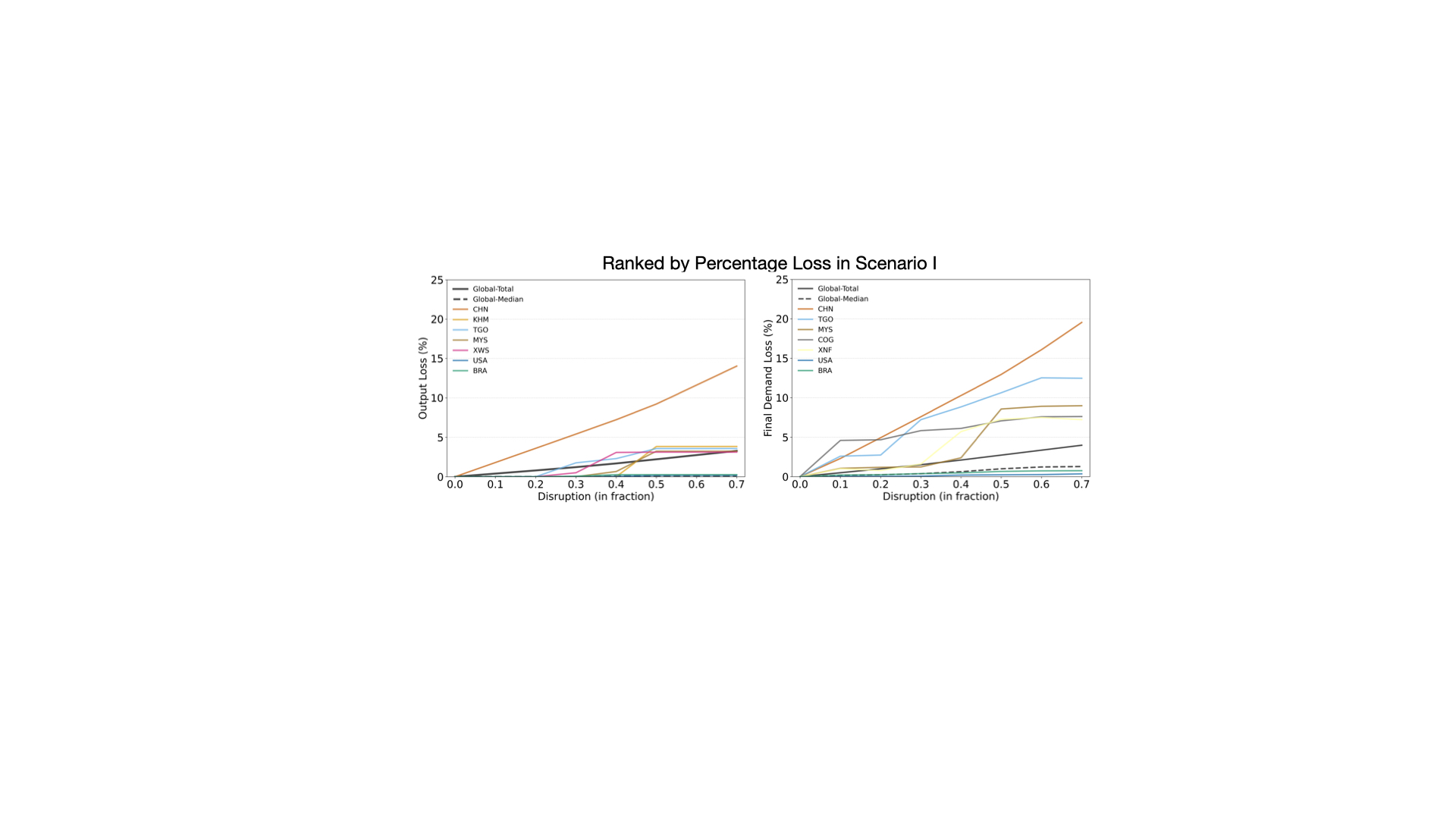}
\caption{\textbf{Progressive cascading impact of USA--China Oilseed export disruption.} While China dominates in output losses, several smaller economies experience disproportionately high final demand losses, indicating heterogeneous exposure to downstream supply-chain effects.}
\label{fig:general_trend_2023_disruption}

\end{figure}

\begin{figure}[t!]
    \centering
    \scriptsize

    \includegraphics[width = \linewidth, clip, trim = 0cm 10cm 0cm 8cm]{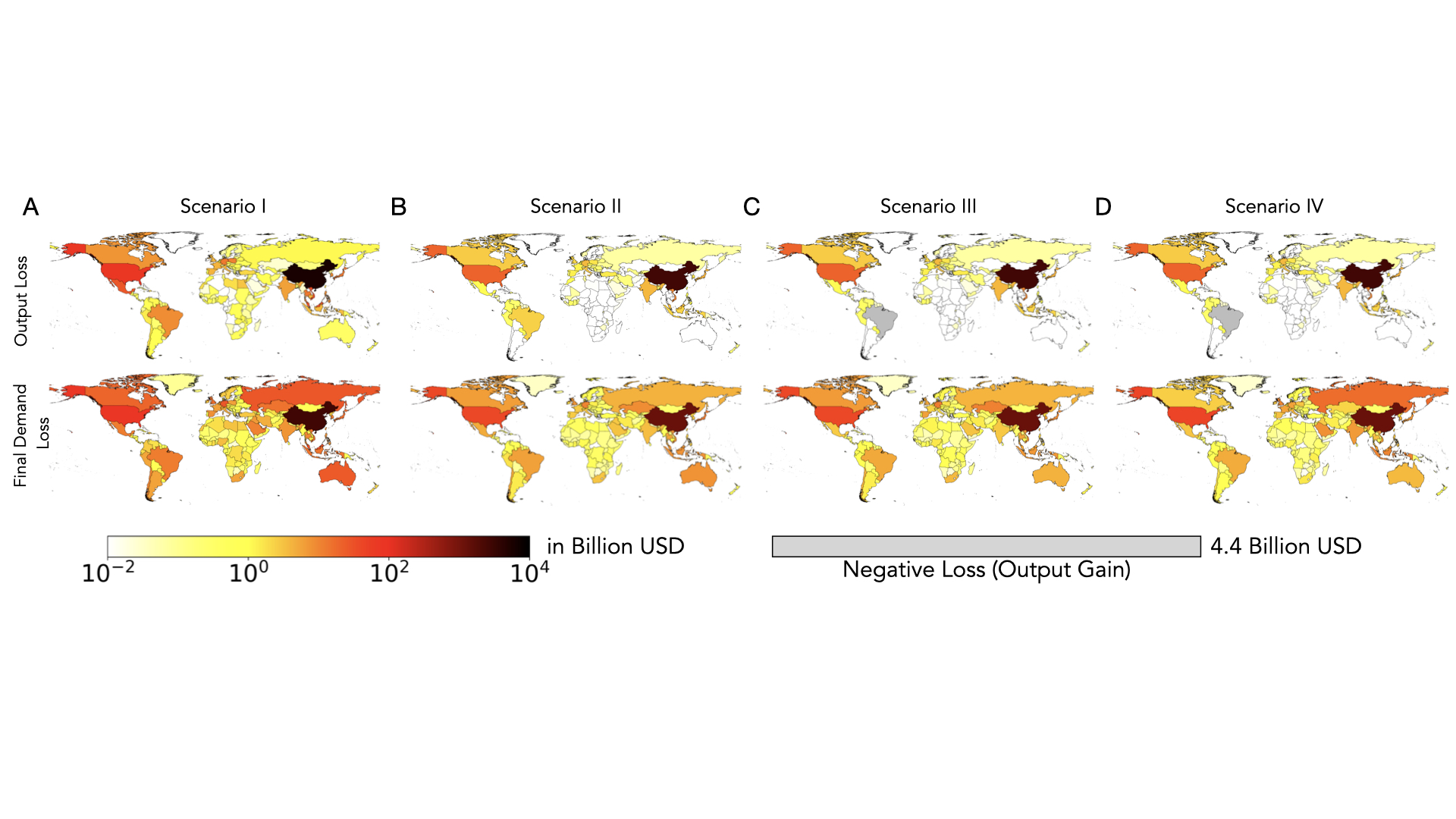}

    \caption{\textbf{Cascading global impacts of USA-China oilseed supply disruption and mitigation scenarios.} Country-level impacts are measured as output loss and final demand loss with respect to the baseline values. In Scenario III and IV, Brazil (shaded in gray) experiences a net negative output loss due to increased oilseeds production.}

    \label{fig:global_impacts}
\end{figure}

\subsection{Impact of Disrupting USA--China Oilseed Trade}

We begin by analyzing the impact of a trade disruption under a baseline scenario in which no endogenous mitigation mechanisms are permitted. In this setting, disrupted trade flows are introduced exogenously by constraining USA oilseed exports to China by a $\Delta$-fraction, while preventing firms from reallocating intermediate inputs across alternative suppliers or expanding production beyond baseline capacity. This scenario mirrors prior IO-based disruption analyses that examine shock propagation under fixed production and trade structures.

\begin{figure}[b!]
\centering
\includegraphics[width = \linewidth, clip, trim = 2cm 3cm 2cm 4cm]{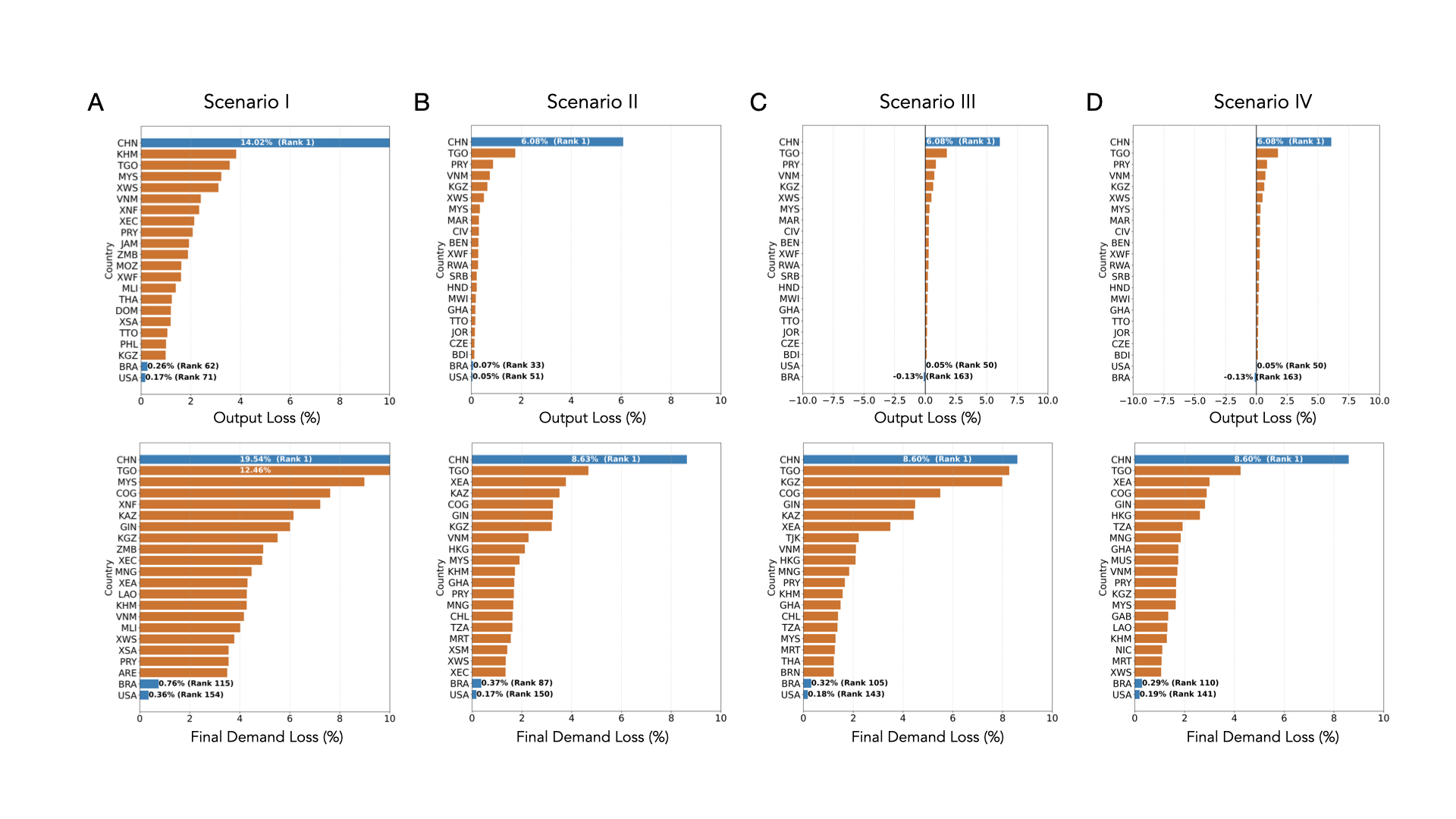}
\caption{\textbf{Percentage country-level impacts across disruption and mitigation scenarios.} Across scenarios, China consistently experiences the largest impact, while mitigation strategies progressively reduce both the intensity and breadth of global spillovers.}
\label{fig:countrywise_ranking_percentage}
\end{figure}

Specifically, we impose disruption parameters $\Delta_{i,j}>0$ on the affected trade links and set the reallocation parameters to $\alpha_{i,j}=1$ for pre-existing flows, while fixing production capacity parameters to $\beta_i=1$ for all firms. As a result, the supply chain responds to the disruption solely through reductions in feasible output, with unmet demand absorbed by the slack variable.

\begin{figure}[t!]
    \centering
    \includegraphics[width = \linewidth, clip, trim = 8cm 0 8cm 0]{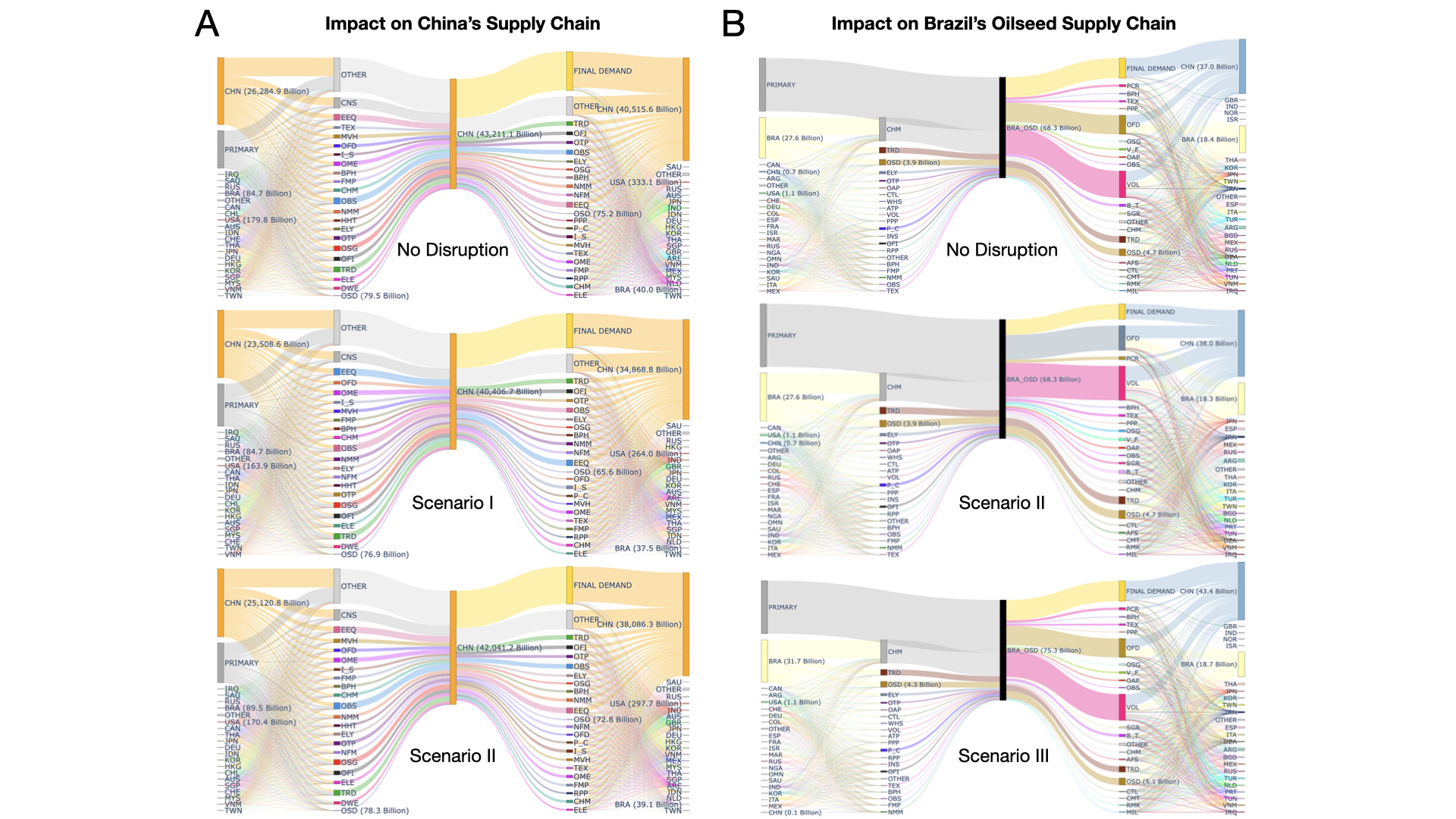}
    \caption{
    \textbf{Supply-chain flow reconfiguration under USA--China oilseed disruption and mitigation.} (A) China's supply chain flows under the baseline no disruption scenario, the disruption-only case (Scenario I), and mitigation through Brazilian oilseed reallocation to China (Scenario II). (B) Brazil's Oilseed supply chain flows under the baseline/no-mitigation case, reallocation to China (Scenario II), and domestic production expansion (Scenario III). Node and link widths are proportional to flow values, with smaller countries and sectors aggregated into ``Other.'' 
}
    \label{fig:sankey_main}
\end{figure}

Figure~\ref{fig:general_trend_2023_disruption} summarizes the country-level impacts of disrupting USA oilseed exports to China under Scenario I. Across all four panels, output and final demand losses increase with the disruption fraction, but the magnitude and timing of these effects vary substantially across economies. The percentage-loss panels reveal strong heterogeneity: while the global median remains low, a small subset of countries experiences much sharper losses, indicating that the aggregate impact is unevenly distributed. China exhibits the steepest increase in absolute losses, consistent with its role as the primary downstream destination of the disrupted flow and therefore the economy most directly exposed to the shock. In contrast, global percentage losses grow more gradually, suggesting that the broader system partially absorbs the disruption through indirect trade adjustment and final-demand reallocation.

A notable pattern emerges in the percentage final-demand-loss panel: at lower disruption levels, several smaller African economies experience disproportionately high relative losses, in some cases exceeding China’s percentage loss. Although their absolute losses remain much smaller, this suggests heightened vulnerability to downstream spillovers through reliance on Chinese firms and sectors affected by the oilseed disruption. By contrast, the median country experiences only modest losses, indicating that severe relative impacts are concentrated among a limited set of vulnerable economies. The figure also shows that final demand losses rise earlier, while stronger output cascades emerge only at higher disruption levels. This suggests that demand-side effects are more immediately sensitive to the shock, whereas production-side cascades require larger disruptions before propagating through intermediate supply chains. Overall, the results reveal nonlinear amplification with disruption severity, a strongly asymmetric burden centered on China, and substantial relative exposure among smaller economies despite lower absolute losses.

In absolute terms, China experiences by far the largest losses, with significant spillovers to major trade-exposed economies such as the United States, Japan, Germany, and several Southeast Asian countries, as shown in Figure~\ref{fig:global_impacts}. Complementing this aggregate view, Figures~\ref{fig:countrywise_ranking_percentage} and~\ref{fig:countrywise_ranking_absolute} report the top 20 regions ranked by percentage and absolute output and final demand losses, respectively.

At the sectoral level, the disruption generates concentrated losses in oilseed-related and downstream production sectors, with China accounting for the dominant share of both output and final demand losses. The oilseed sector (highlighted by red box) is directly affected, but the impacts extend beyond to related sectors, as illustrated in Figures~\ref{fig:sectorwise_ranking_percentage} and ~\ref{fig:sectorwise_ranking_absolute}. In percentage terms, several sectors exhibit sizable losses despite relatively small contributions from other countries, indicating that the shock remains strongly centered on China but propagates unevenly through sectoral linkages. Final demand losses are concentrated in a smaller set of sectors, with some sectors showing large contributions from ``Others,'' suggesting that downstream impacts spill over more broadly across the global economy. 

This baseline scenario serves as a reference point for evaluating mitigation mechanisms. In the following sections, we examine how endogenous trade reallocation and production expansion alter these outcomes, and whether such mechanisms reduce cascading losses while introducing new pressures elsewhere in the system.

\subsection{Redirecting Brazil's Oilseeds to reduce the cascading impacts}\label{sec:Scenraio2}

\begin{figure}[b!]
    \centering
    \includegraphics[width=\linewidth, clip, trim = 0cm 18cm 0cm 2cm]{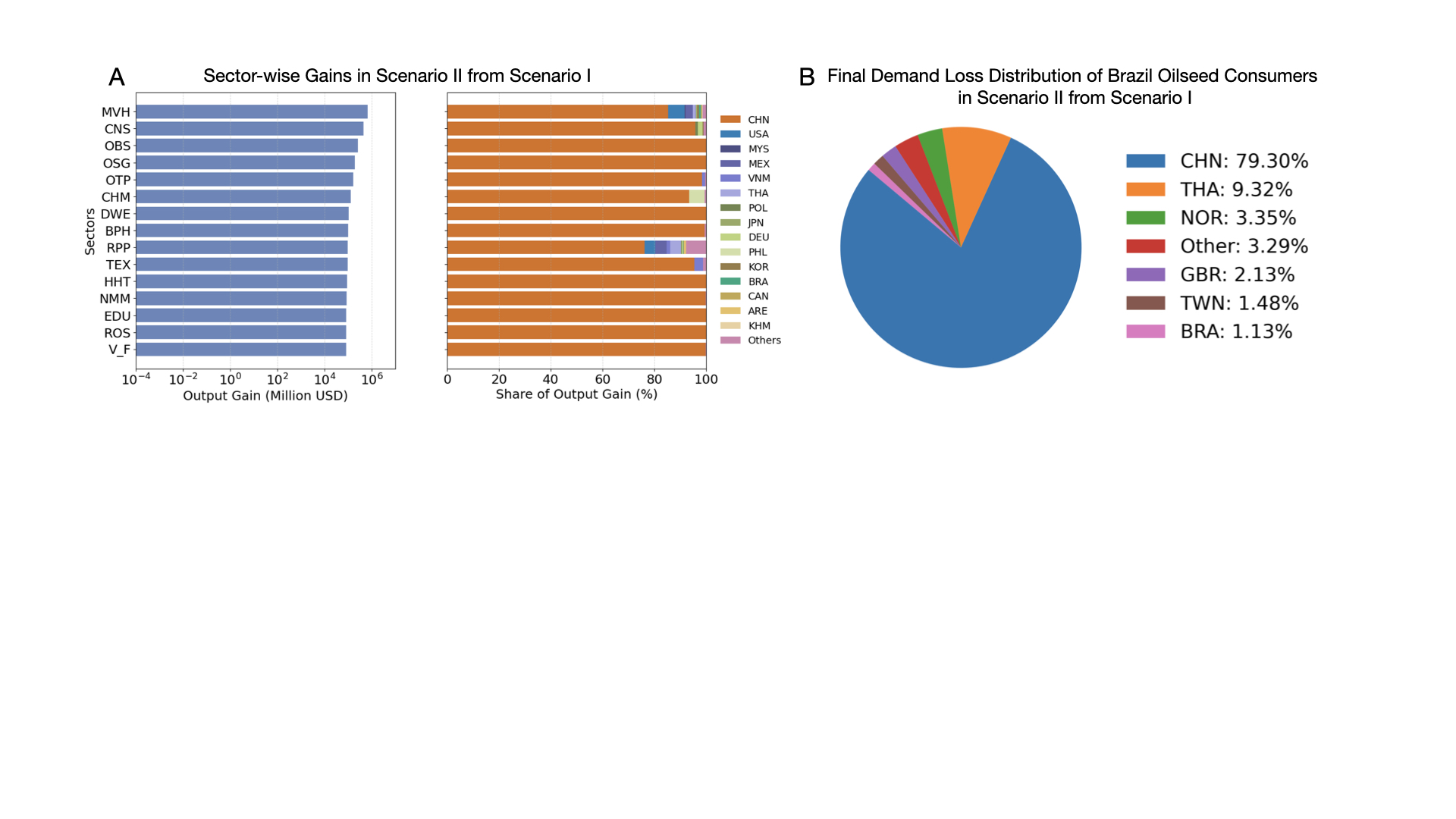}
    \caption{\textbf{Reallocation of Brazil's oilseed exports to China under Scenario II relative to Scenario I.} (A) Top 13 sectoral gains in absolute terms and corresponding producer-region shares induced by reallocation. (B) Distribution of reallocation-induced final demand losses among countries importing oilseeds from Brazil, showing that China absorbs the majority of the indirect impact, followed by Thailand and Norway.}
    \label{fig:scenario2_wrt_1}
\end{figure}

As a mitigation strategy, we examine whether reallocating oilseed exports from an alternative major producer can reduce the cascading impacts of the USA--China oilseed disruption. We focus on Brazil, the second-largest exporter of oilseeds to China, as shown in Figure~\ref{fig:USA_OSD_flows}. To model this counterfactual, we increase the permitted flows from Brazil's OSD sector to China by an additive $\alpha$ factor, while all other flows remain capped at their baseline MRIO values ($\alpha=1$). This isolates the extent to which a single major oilseed exporter can offset the loss of U.S. supply through reallocation alone, while also capturing potential losses imposed on other global consumers of Brazilian oilseeds as supply is redirected toward China.

Figure~\ref{fig:global_impacts}(B) and Figure~\ref{fig:countrywise_ranking_percentage}(B) illustrate the LP solution with a $20\%$ reallocation of Brazilian oilseed exports to China under a $70\%$ disruption of USA oilseed exports to China. We observe that losses are markedly reduced and become less widespread compared with the disruption-only case, with percentage losses falling for most countries and the largest mitigation observed for China. The set of severely affected countries also narrows, indicating that reallocated Brazilian supply alleviates the input bottleneck created by the U.S. shock to Chinese sectors. At the same time, final demand losses do not vanish entirely, reflecting the LP's trade-off structure: final demand flows are redirected to relieve output constraints, leading to smaller but still concentrated residual impacts. Interestingly, Brazil also experiences a reduction in net loss, from $0.76\%$ to $0.37\%$, suggesting that this reallocation can be mutually beneficial in aggregate.

At the sectoral level, Figure~\ref{fig:scenario2_wrt_1}(A) shows that reallocation generates substantial output recovery across most sectors relative to the disruption-only scenario, reducing cascading impacts below the levels induced by the oilseed shock. The largest gains are concentrated in Chinese sectors, including the OSD sector, while spillover gains to other regions remain below $20\%$. Similar country-level patterns appear for final demand gains, although the OSD sector does not appear among the top affected final-demand sectors. The largest recovery is observed in China's motor vehicles and parts sector (MVH), at approximately $70\%$, consistent with MVH dropping out of the top disrupted sectors in Figure~\ref{fig:sectorwise_ranking_percentage}(B). Exceptions include sectors such as textiles (TEX) and rubber and plastics (RPP), where residual losses remain more persistent.

Finally, Figure~\ref{fig:scenario2_wrt_1}(B) highlights the redistribution of final-demand losses among countries that consume Brazilian oilseeds. China absorbs the majority of the reallocation-induced final-demand impact, accounting for nearly $80\%$ of the loss, followed by Thailand and Norway. This indicates that while Brazilian reallocation substantially mitigates cascading losses from the original U.S. disruption, it also shifts part of the burden onto countries connected to Brazil's oilseed demand network.

\subsection{Increasing Brazil's Oilseed Production through Domestic Reallocation}
\label{sec:scenario3}

Establishing new international trade agreements often requires navigating complex geopolitical constraints and existing treaties, making such adjustments time-consuming and uncertain. To assess more immediate mitigation strategies, we next analyze interventions within Brazil's domestic economy.

Specifically, we examine the impact of increasing Brazil's oilseed production through internal reallocation of resources. To support this expansion, domestic Brazilian firms are allowed to redirect supply toward Brazil's oilseed sector, enabling a targeted production increase of $\eta = 0.10$. In addition, we assume a proportional expansion in primary inputs available to the oilseed sector, with $\beta = 0.10$. This setup isolates how domestic adjustments affect Brazil's production capacity and the broader global supply chain.

Figure~\ref{fig:global_impacts}(C) illustrates the LP solution when Brazil's oilseed production is increased by $10\%$ through domestic reallocation. We observe negative output-loss values, shaded in gray, indicating a net increase in the total value of goods and services produced relative to the non-disrupted baseline. This corresponds to a net gain of $0.13\%$ for Brazil, as shown in Figure~\ref{fig:countrywise_ranking_percentage}(C). However, the increase in Brazilian oilseed production also raises percentage final-demand losses for some smaller economies, suggesting that domestic reallocation places additional pressure on countries indirectly connected to Brazil's redirected supply flows. In contrast, most large economies show minimal changes in absolute output losses relative to Scenario II. Final-demand losses, however, are redistributed across larger economies, with many experiencing a reduction in losses, as illustrated in Figure~\ref{fig:sectorwise_ranking_absolute}.

At the sectoral level, Figure~\ref{fig:sectorwise_domestic_absolute} shows that increasing Brazil's oilseed production produces a mixed redistribution of losses and gains relative to Scenario II. Final demand losses are distributed acccross sector with major disruptions to chemicals (CHM), vegetable oils and fats (VOL), trade (TRD), and textiles (TEX) sectors. These losses are largely associated with Brazil redirected domestic inputs, which not only impacts Brazil's final demand but indirectly impact other regional production. At the same time, the right panels show sizable final demand gains in oilseeds (OSD) and related sectors, with Brazil accounting for a dominant share of oilseed gains. Overall, the sectoral patterns suggest that domestic production expansion can mitigate the original international disruption, but it does so by reallocating pressure across Brazil's internal production network and its downstream trade partners.

\begin{figure}[t!]
    \centering
    \includegraphics[width = \linewidth, clip, trim = 0cm 11cm 0cm 9cm]{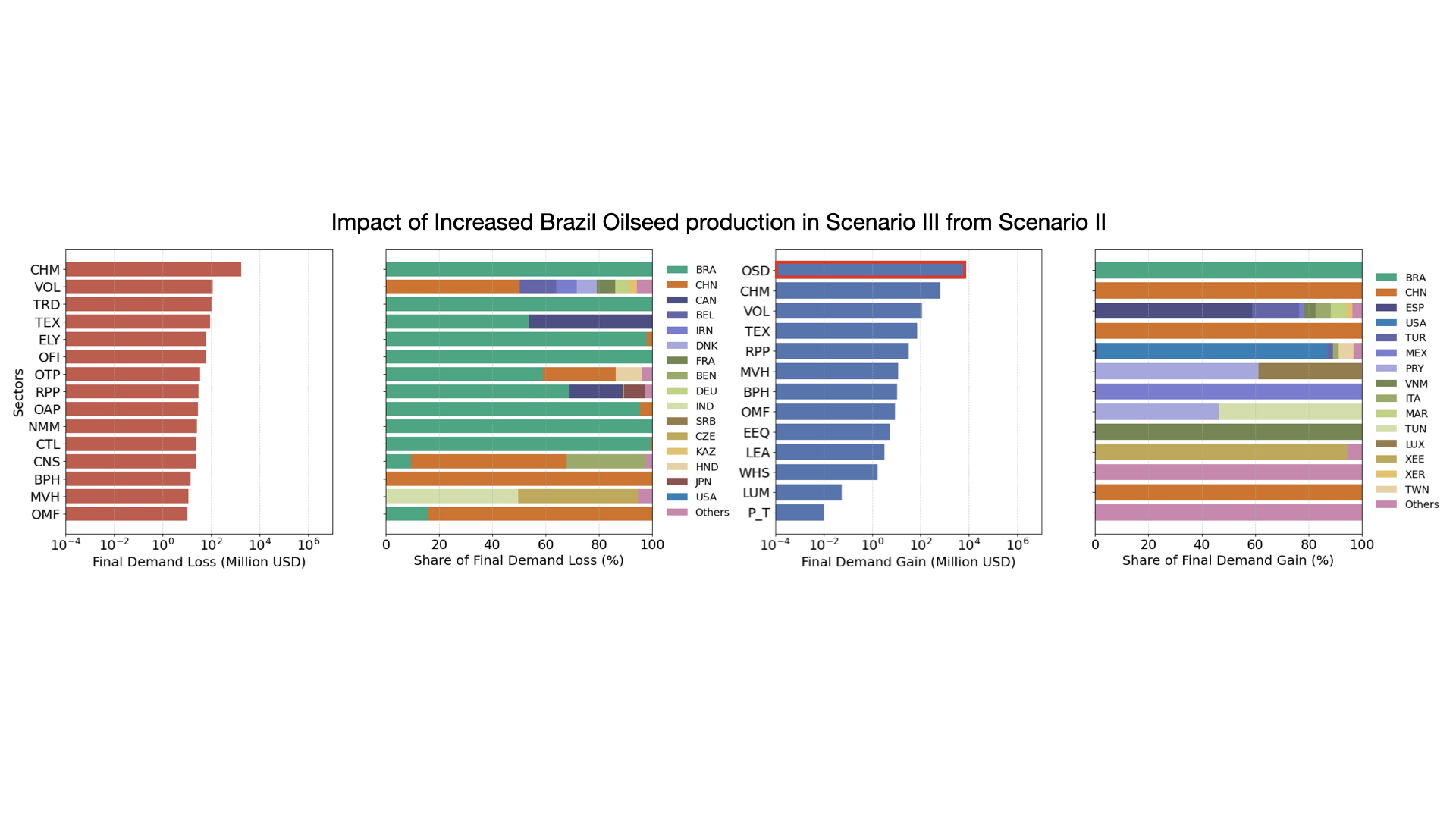}
    \caption{\textbf{Absolute sector-level impact of increasing Brazil's oilseed production under Scenario III relative to Scenario II.} The panels show the top impacted sectors and their corresponding producer region shares.}
    \label{fig:sectorwise_domestic_absolute}
\end{figure}

\subsection{Expanding Brazil's Oilseed Production through Global Partnerships}

As observed earlier, domestic reallocation to increase Brazil's oilseed production disproportionately affects smaller economies across Africa and Asia. As an alternative mitigation strategy, we consider a counterfactual scenario in which Brazil's oilseed sector can draw the required inputs from international trade partners, rather than relying only on domestic sectors. To support this expansion, all firms are allowed to redirect supply toward Brazil's oilseed sector, enabling a targeted production increase of $\eta = 0.10$, along with a proportional expansion in primary inputs, $\beta = 0.10$.

Figure~\ref{fig:countrywise_ranking_absolute}(D) shows that Scenario IV produces country-level output losses that are nearly identical to Scenario III, indicating that global partnership-based production expansion provides only marginal improvement beyond Brazil's domestic reallocation strategy in terms of aggregate output. The main difference appears in the redistribution of absolute final demand losses: the ranking of affected economies shifts, but the overall magnitude remains largely unchanged, except for Russia and Kazakhstan. Russia becomes the third-largest affected economy in Scenario IV, while Kazakhstan drops out of the top 20 after being among the most affected countries in Scenario III. Since absolute losses are naturally skewed toward larger economies, Figure~\ref{fig:countrywise_ranking_percentage} provides a complementary view, showing that several smaller economies experience larger percentage recoveries under Scenario IV. This suggests that diversifying input sources provides particularly strong relative benefits for smaller economies, even when aggregate global losses change only modestly.

\begin{figure}[t!]
\centering
\includegraphics[width = \linewidth, clip, trim = 0cm 12cm 0cm 9cm]{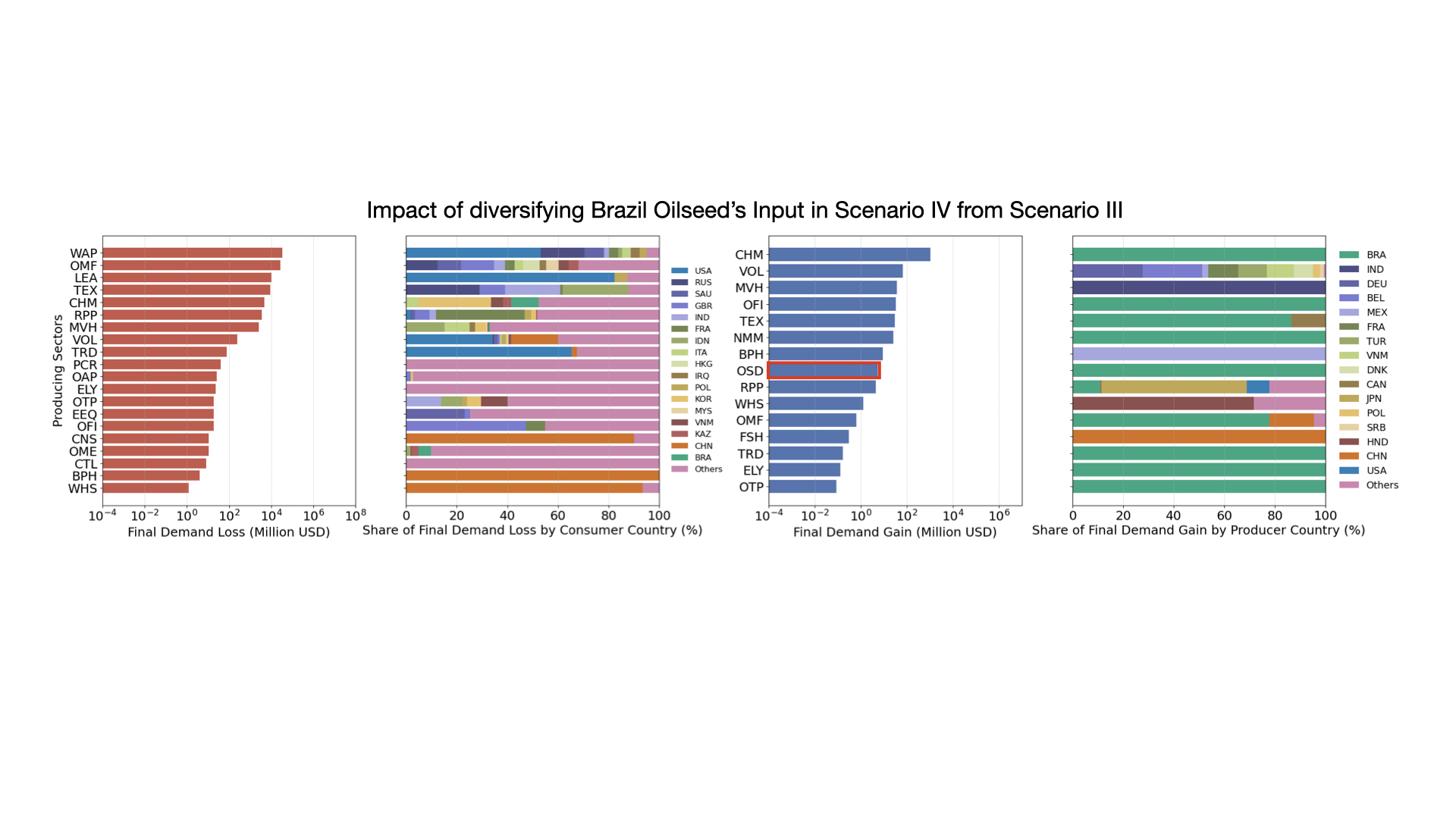}
\caption{\textbf{Absolute sector-level impact of diversifying Brazil's oilseed inputs under Scenario IV relative to Scenario III.} The panels report impacted sectors and their corresponding consumer- and producer-country shares.}
\label{fig:sectorwise_global_absolute}
\end{figure}

While the global sector-level aggregate loss remains largely unchanged between Scenarios III and IV, as illustrated in Figure~\ref{fig:sectorwise_ranking_percentage}, Figure~\ref{fig:sectorwise_global_absolute} highlights the relative redistribution induced by diversifying Brazil's oilseed inputs. Compared with the domestic reallocation strategy in Scenario III, Scenario IV spreads input pressure across a broader set of producer countries, generating final demand gains in oilseed and related sectors while reducing Brazil's exclusive burden of supporting expanded oilseed production. The country-share panels show that losses are no longer concentrated only within Brazil's domestic economy, but are distributed across a wider set of trade partners. At the same time, the final-demand gain panel shows that Brazil captures the largest share of gains across most sectors, indicating that global input diversification primarily strengthens Brazil's expanded oilseed production while reallocating residual adjustment costs across the international supply network.

\section{Conclusion}

We developed a linear programming framework for analyzing cascading supply-chain disruptions under trade shocks. The framework extends standard IO-based disruption analysis by jointly modeling disruption propagation, trade reallocation, production expansion, and capacity constraints within a unified equilibrium formulation. This enables counterfactual analysis of how disrupted flows are absorbed when substitution is constrained by the broader production network.

We applied the framework to USA-China oilseed trade disruptions using GTAP MRIO data. The results show that disruptions to U.S. oilseed exports to China generate uneven impacts across countries and sectors: China bears the largest absolute losses, while several smaller economies experience high relative final-demand losses through downstream exposure. Thus, trade shocks propagate beyond the directly affected bilateral partners through intermediate production and final-demand channels.

We also examined mitigation through Brazilian oilseed reallocation and production expansion. Reallocating Brazilian exports to China substantially reduces cascading losses, but shifts pressure onto countries and sectors linked to Brazil's oilseed demand network. Increasing Brazilian production further offsets the disruption, while relying on global inputs spreads adjustment costs more broadly, with modest aggregate gains relative to domestic expansion.

Together, these findings highlight a central trade-off in global supply-chain resilience: trade reallocation and production expansion can reduce aggregate losses, but their effectiveness depends on spare capacity, existing trade relationships, and where adjustment costs are absorbed by final demand. Trade diversion is therefore conditional rather than automatic; under constrained substitution and production capacity, mitigation can reduce cascading losses while creating localized resource stress.

\bibliographystyle{unsrt}
\bibliography{references}

\appendix
\clearpage 

\begin{center}
{\Large\textbf{Appendix}}
\end{center}

\section{Background \& Related Work}

In this section we discuss the relevant background and techniques used in our work while putting  it in  context with the existing literature.

\smallskip 

\noindent \textbf{Input–Output (IO) System.} At its core, IO analysis combines an economic model of production with empirically constructed IO tables that record the flows of goods and services among industrial sectors and geographical regions. We describe these below. 

\smallskip

\noindent The classical Leontief input–output (IO) model \cite{leontief1986input} characterizes the global supply chain as an interdependent production network, where the output of one sector serves as an essential input to others. The model is defined by two key components:

First, the \textit{Leontief production function} constrains a firm’s output by the input that is least available relative to its required input coefficients. Let $X$ denote the output of the firm, and let $S_k$ be the available intermediate input from sector $k$ with input coefficient $A_k$, and $P_0$ and $B_0$ denote the available primary input and its corresponding coefficient. Then, 
 \begin{align}\label{eq:production_function}
    \textstyle X = 
    \min \left(
        \frac{S_{0}}{A_{0}},
        \frac{S_{1}}{A_{1}},
        \dots,
        \frac{P_{0}}{B_{0}}
    \right).
\end{align}

Second, the \textit{IO equilibrium} condition requires that the total outflow of a firm, which includes intermediate supplies to other firms and final demand, does not exceed the firm’s production level. This ensures consistency between feasible production and feasible flow of goods across the supply chain network.

\smallskip

\noindent Next, key ingredient of IO Analysis are the \textit{Input-Output Tables} used to record the flow of economic activities among different sectors of the supply chain. They are developed empirically for economic research and structural analysis at the global scale, highlighting inter-industrial relationships covering various sectors of the economy. In this work, we use a \textit{Multi-Region Input-Output} (MRIO) table from \cite{peters2011constructing}, as the basis for our analysis. The system consists of firms belonging to various sectors, $\mathcal{S}$ in the global economy consisting of multiple regions, $\mathcal{R}$. The MRIO table provides intermediate flow values (in millions of USD) among various region-sector firm pairs, the primary inputs consumed by each firm and the household demands from each firm in the system. Formally, suppose $\mathcal{F}$ is the set of firms in the system, then let $Z_{i,j}^*$ denote the intermediate flow from firm $i$ to firm $j$, $H_{i}$ denote the total value of household demand from firm $i$ and $P_i$ denote the value of primary inputs consumption of firm $i$. Then, for a firm $i$, we can estimate the output, denoted by $X_i^*$, the sector-wise Input coefficients for each sector $s$, $A_{s,i}$ and Primary input coefficients, $B_i$ for the Leontief production function as:
\begin{align}\label{eq:ss_output}
    X_i^* &= \textstyle \sum_{j\in \mathcal{R} \times \mathcal{S}} Z^*_{i,j} + H_i\,; \quad A_{s,i} \; = \; \frac{\sum_{j \in \mathcal{R} \times \{s\}}Z^*_{j,i}}{X_i^*}, \quad \text{and} \quad  B_{i} \; = \; \frac{P_i}{X_i^*},
\end{align}
where $\mathcal{R} \times \{s\}$ denotes the set of firms in the system associated with sector $s$ across all regions. These estimated coefficients are crucial in capturing the Leontief production rule for each firm.

Together, these components form the foundation of IO analysis for studying global supply chains. 
Miller and Blair \cite{miller2009input} provide a comprehensive treatment of the development and applications of IO models. In this work, we study the impact and mitigation strategies for tariff-induced disruptions.

\smallskip
\noindent \textbf{Cascading Disruption in Global Supply Chains.} 
Due to the highly interconnected nature of modern production networks, localized disruptions can propagate through intermediate input dependencies and give rise to cascading effects across regions and sectors. The study of cascading failures in supply chains has therefore received significant attention in operations research, economics, and network science. Comprehensive surveys of both classical and recent work can be found in \cite{snyder2016or,korder2024simulation}.

Recent studies have increasingly relied on simulation-based approaches to examine disruption propagation. For example, Guan et al.~\cite{guan2020global} study the cascading impacts of COVID-19 lockdowns on global supply chains using an agent-based model (ABM) built on an IO representation. Other works explore supply-chain robustness in the presence of cyber-physical coupling~\cite{mu2021robustness}, resilience of ports in shipping network~\cite{cao2025data}, and resilience of major cities to food-supply shocks \cite{gomez2021supply} using ABMs and network-based simulations that incorporate firm-level heterogeneity and behavioral adaptation. These models provide valuable insights into localized dynamics and scenario-specific responses. 

However, ABM-based approaches typically require extensive calibration and rely on behavioral assumptions governing agent responses to shocks. Moreover, their computational complexity limits scalability, making it challenging to apply them to economy-wide systems involving millions of inter-firm flows. As a result, such models do not naturally yield global equilibrium characterizations or transparent system-level guarantees, which are essential for analyzing large-scale trade policy interventions.

\medskip
\noindent \textbf{Modeling Supply Chains Using Linear Programs.} 
An alternative and complementary line of work formulates supply-chain dynamics as optimization problems, enabling system-level analysis under explicit feasibility and capacity constraints (See~\cite{eiselt2007linear} for a detailed overview). Linear programming (LP) formulations have been widely used to study network flows problem~\cite{bazaraa2011linear}, resource allocation problems~\cite{panik2018linear}, network interdiction problems~\cite{woodruff2006network} and supply-chain optimization~\cite{spitter2005linear}.

In the context of IO systems, LP formulations has become a textbook technique for supply-chain optimization~\cite{vogstad2009input}. Given an objective function, LP models enable choosing the level of activities to optimize, satisfying the sectoral level relations imposed by IO analysis. These yield equilibrium states of the production network without relying on behavioral assumptions, which is unlike simulation-based approaches. Existing works employ these techniques to study energy-environment tradeoffs in the global economy~\cite{oliveira2016coupling}, and impact of disruptions to interconnected economies~\cite{santos2006inoperability}. Recent work by Soltanisehat et al. ~\cite{soltanisehat2024multiregional} proposes a novel multiobjective mixed-integer linear programming formulation for optimizing lockdowns policy for the US economy during COVID-19. In this work we ask similar questions for modeling the supply-chain disruptions induced due to changing geopolitical dynamics in the form of tariff imposition. 

Specifically, we build on these optimization-based perspectives and formulate the MRIO system as a linear program to analyze the cascading impact of tariff-induced disruptions. Our approach enables the endogenous reallocation of trade flows and adjustment of production levels in response to shocks, allowing us to study both the immediate propagation of disruptions and the mitigating role of alternative trade pathways. This framework is particularly well suited for global supply-chain analysis, where policy-induced shocks generate competition for scarce inputs across regions and sectors and require coordinated reallocation across the entire production network. 

\smallskip

\noindent\textbf{Solving the IO model using Linear Programming.} We first describe a linear programming (LP) approach for solving the IO model~\cite{vogstad2009input}.
We later extend this to handle different kinds of disruptions and relevant mitigation strategies.

We treat each region–sector pair as a distinct firm $i \in \mathcal{F}$ and captures how production levels, input requirements, and reallocation constraints
interact under both normal and disrupted conditions. 
For each firm $i$, the variable $x_i \ge 0$ denotes its total output, while $z_{i,j} \ge 0$ represents the intermediate flow supplied from firm $i$ to firm $j$. The LP has the following structure.

\begin{eqnarray}
\nonumber
\max && \sum\nolimits_{i \in \mathcal{F}} x_i, \mbox{ subject to}\\
A_{s,i} x_i 
        &\leq& 
\sum\nolimits_{\substack{j \in \mathcal{F}: \sigma(j)=s}} z_{j,i}, \ \forall s \in \mathcal{S}, \forall i\in\mathcal{F} \label{eqn:leontief1}\\
B_i x_i &\le& P_i, \forall i\in\mathcal{F} \label{eqn:leontief2}\\
z_{i,j} 
        &\le& Z^*_{i,j}, \forall i, j\in\mathcal{F} \label{eqn:cap}\\
\sum_{j\in\mathcal{F}}z_{i,j} + H_i 
        &\le& x_i, \forall i\in\mathcal{F} \label{eqn:cap_demand_io}\\
\nonumber
x_i, z_{i,j} &\geq& 0, \forall i, j
\end{eqnarray}

The constraints (\ref{eqn:leontief1}) and (\ref{eqn:leontief2}) capture the Leontief production constraints; Observation \ref{obs:baseLP} shows that the maximization objective ensures that the optimum LP solution satisfies the Leontief constraints exactly.

\begin{observation}
\label{obs:baseLP}
The optimal solution $x^*, z^*$ to the above LP satisfies the IO model.
The MRIO table is a feasible solution to the LP.
\end{observation}
\begin{proof}
Since the variables are all non negative,
the maximization objective ensures that 
\[
x^*_i = \min\{\min_{j \in \mathcal{F}}\{ \sum\nolimits_{ \sigma(j)=s} z^*_{j,i}/A_{s,i}\}, P_i/B_i\};
\]
if $x^*_i$ is strictly less than the minimum, it can be increased, which would violate optimality.
This corresponds to the Leontief constraints, and  implies that the LP solution satisfies the IO model.

Further, the constraint $z^*_{ij}\leq Z^*_{ij}$ ensures that the MRIO table is a feasible solution.
\end{proof}

\newpage
\section{Additional Empirical Details}

\subsection{Overview of GTAP-MRIO.}

The GTAPv12 database captures trade flows among 163 economic regions (corresponding to set $\mathcal{R}$) and 65 sectors (corresponding to set $\mathcal{S}$) for each region. We construct the multi-region input–output (MRIO) table from the GTAPv12 (2023) release following the methodology of \cite{peters2011constructing,zheng2021chinese}. Consistent with our modeling approach in Section~\ref{sec:method}, each \emph{region–sector} pair defines a firm, and each region has an associated single representative household. Additionally, each firm has a primary input capturing the economic flow relevant to labor, capital and natural resources directly used by the firm. 

The resulting GTAP-MRIO table consists of $10595$ firms, $6.87$ million intermediate flows between firms, $163 \times 10595$ demand flows to households, and $10595$ primary input flows to firms. The intermediate flows of the MRIO correspond to $Z^*$, the household flows to $H$ and the primary inputs to $P$ of the base model in Section \ref{sec:method}.

\smallskip

\noindent \textbf{Sparsifying the GTAP-MRIO} While the MRIO table provides a detailed representation of intermediate flows across the global supply chain, solving the associated linear program is computationally expensive due to the scale of the data. To mitigate this burden, we construct a sparsified version of the baseline GTAP-MRIO that preserves the most economically significant inter-firm dependencies while reducing the dimensionality of the problem.

\begin{table}[h!]
\centering
\scriptsize
\caption{Impact of value-based sparsification on the MRIO intermediate flow matrix.}
\resizebox{0.85\textwidth}{!}{
\begin{tabular}{|c|c|c|c|c|}
\hline
\textbf{Threshold, $\tau$} 
& \textbf{Aggregate(\%)} 
& \textbf{Mean $\pm$ Std (\%)} 
& \textbf{Median (\%)} 
& \textbf{Intermediate Flows} \\
\hline
$1\times10^{-6}$  & $5.622\times 10^{-6}$    & $0.007 \pm 0.09$   & $0.0001$     & 40,159,151 \\
$1\times10^{-4}$  & 0.0005    & $0.167 \pm 1.03$   & 0.010    & 20,190,035 \\
$\mathbf{1\times10^{-2}}$  & $\mathbf{0.0311}$    & $\mathbf{2.864 \pm 9.32}$   & $\mathbf{0.386}$   & $\mathbf{6,160,324}$  \\
$1$               & 0.8928    & $21.32 \pm 30.28$ & 7.490   & 872,148    \\
$10^{2}$          & 10.35   & $69.13 \pm 36.57$ & 100.0      & 53,713     \\
$10^{4}$          & 55.80   & $98.11 \pm 10.61$ & 100.0      & 1,089      \\
$10^{6}$          & 100.0  & $100.0 \pm 0.0$    & 100.0      & 0          \\
\hline
\end{tabular}
}
\label{tab:coefficients_tau}
\end{table}

\begin{figure}[h!]
    \centering
    \begin{minipage}{0.45\linewidth}
        \centering
        \includegraphics[width=\linewidth]{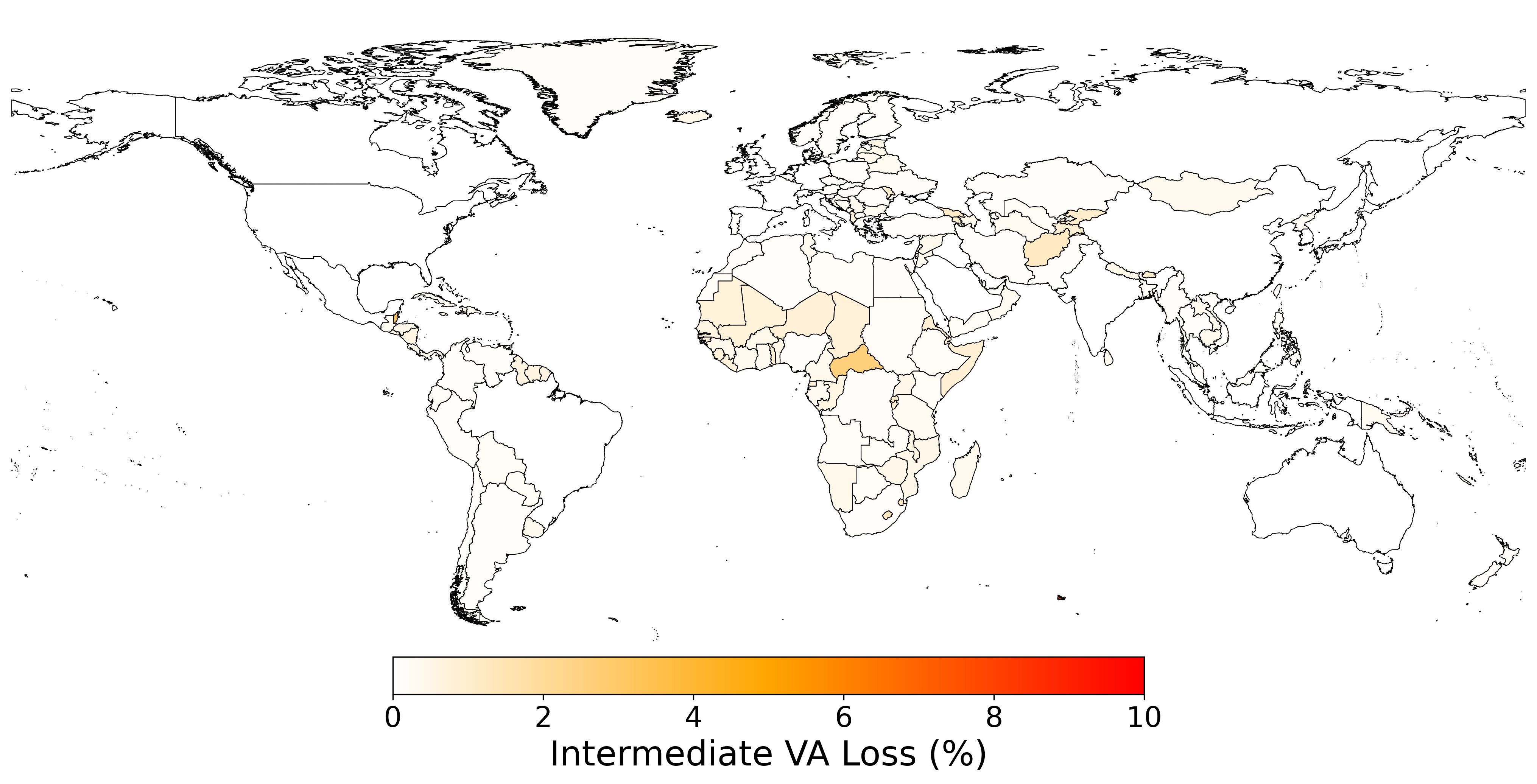}
    \end{minipage}
    \begin{minipage}{0.45\linewidth}
        \centering
        \includegraphics[width=\linewidth]{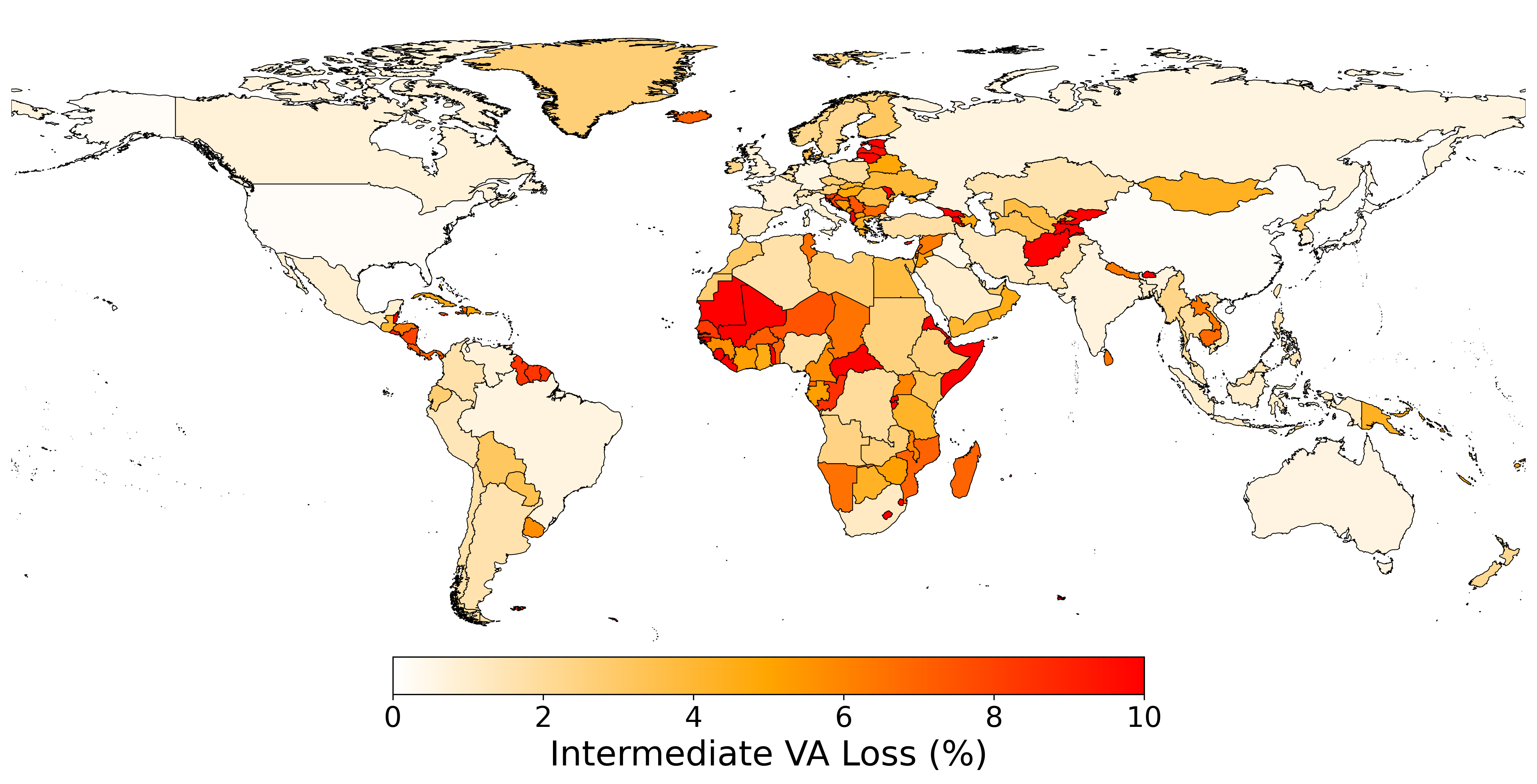}
    \end{minipage}
    \caption{Spatial distribution of global Value-Added loss at threshold (Left) $\tau = 0.01$ and (Right) $\tau = 1.0$.}
    \label{fig:MRIO_tau_heat}
\end{figure}

We evaluate the impact of sparsification for various threshold values of $\tau \in \{ 10^{-6},$ $10^{-4}, 10^{-2},10^{0},10^2,10^4, 10^6\}$ measured as the . We summarize the impact of sparsification on the MRIO table in Table~\ref{tab:coefficients_tau}, and visualize the global distribution of VA Loss in Fig.~\ref{fig:MRIO_tau_heat}.  We observe that as $\tau$ increases, more intermediate flows are removed, leading to rising value-added losses and, at high thresholds, near-systemic collapse. We set $\tau = 0.01$ for the remainder of this work.

\subsection{Supporting Empirical results}

\begin{figure}[h!]
    \centering
     \includegraphics[width = 0.9\linewidth, clip, trim = 2cm 2cm 2cm 4cm]{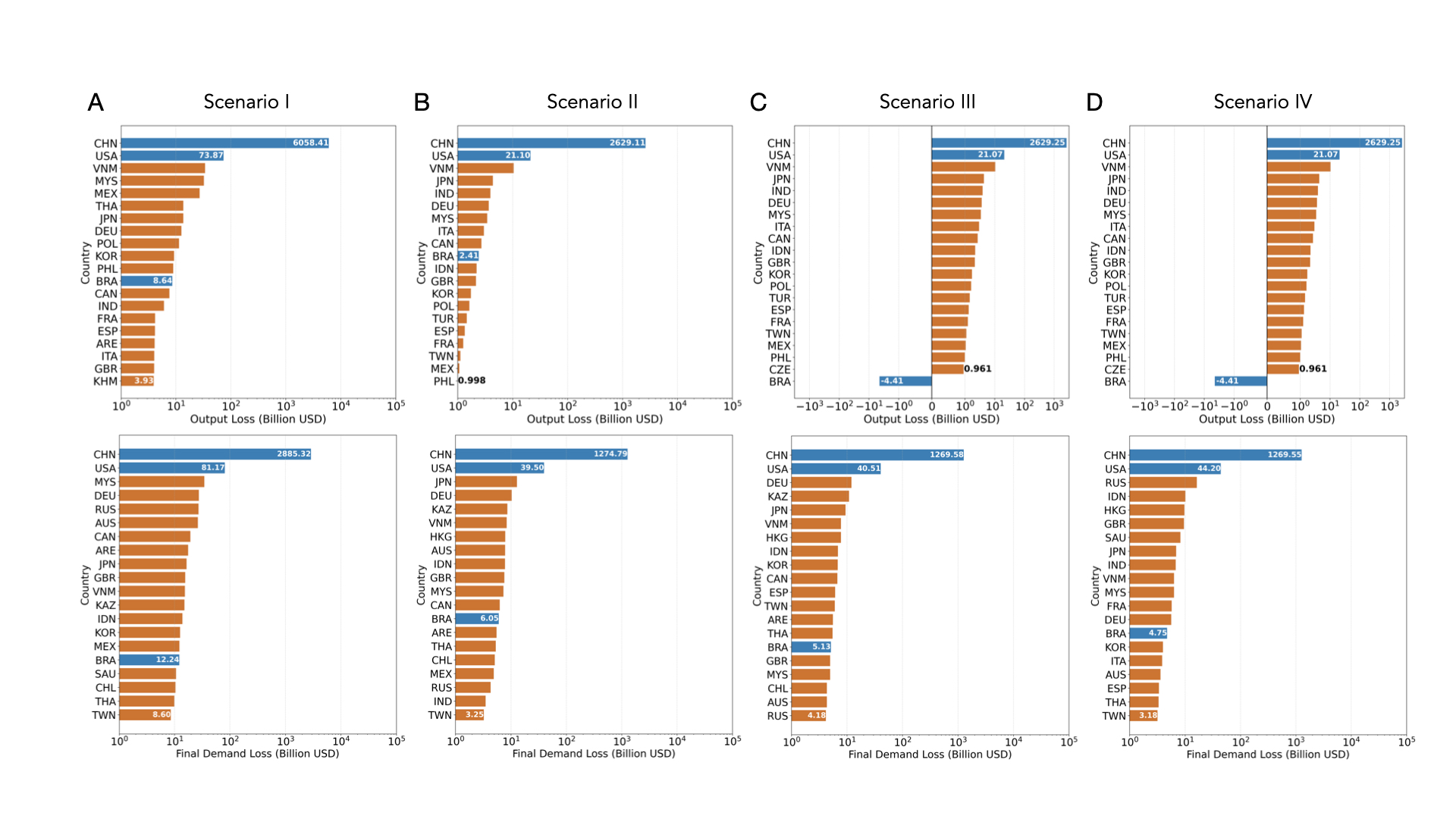}
    \caption{\textbf{Absolute country-level impacts across disruption and mitigation scenarios.} Across scenarios, China consistently experiences the largest impact, while mitigation strategies progressively reduce both the intensity and breadth of global spillovers.
}
    \label{fig:countrywise_ranking_absolute}
\end{figure}

\begin{figure}[h!]
    \centering
    \includegraphics[width = 0.9\linewidth, clip, trim = 2cm 4cm 0cm 4cm]{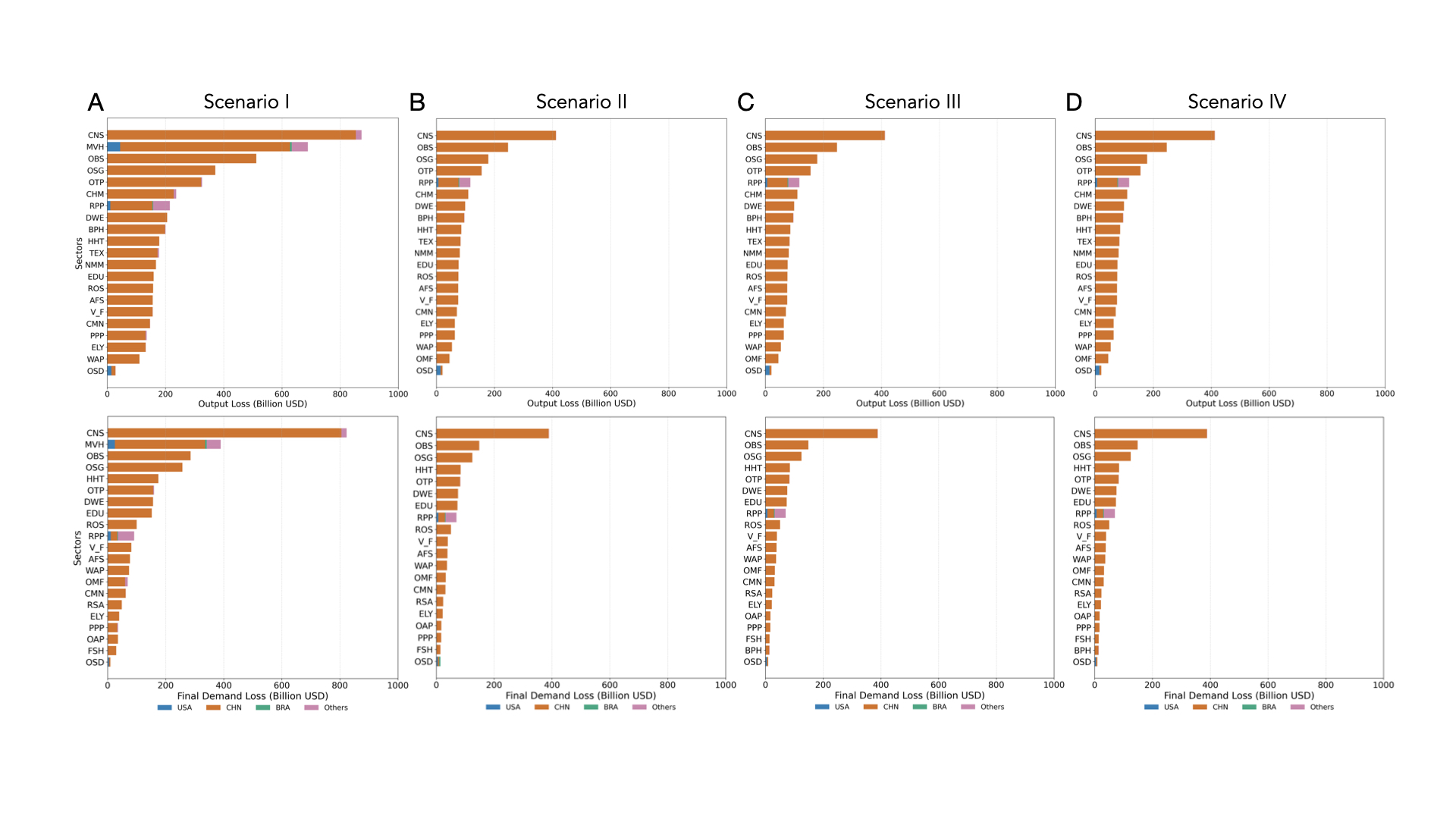}
    \caption{\textbf{Absolute global sector-level impacts across disruption and mitigation scenarios.} Top panels report output loss (\%), and bottom panels report final demand loss (\%) across sectors under four scenarios.
}
    \label{fig:sectorwise_ranking_absolute}
\end{figure}

\begin{figure}[h!]
    \centering
    \includegraphics[width = \linewidth, clip, trim = 2cm 2cm 0cm 4cm]{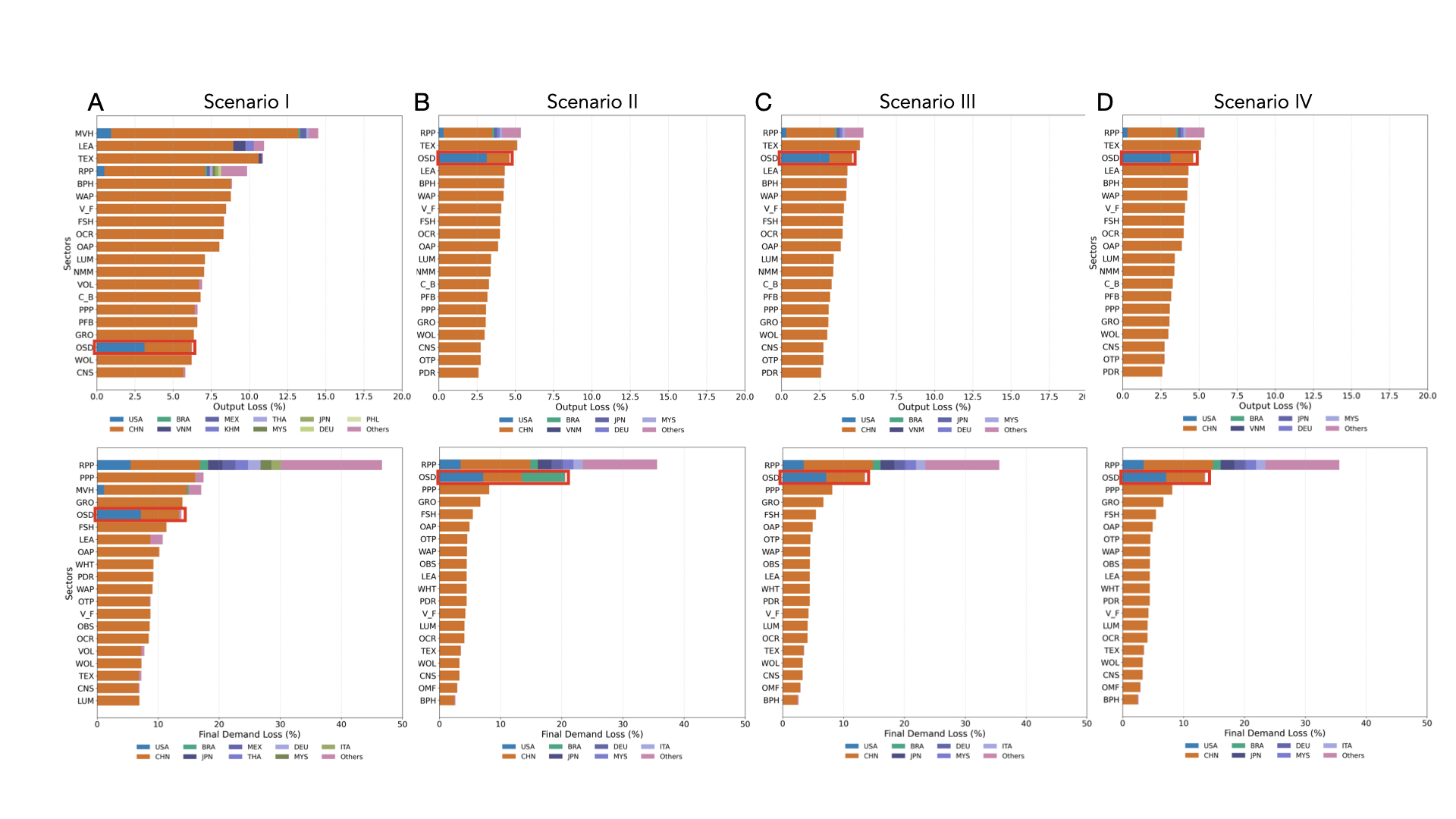}
    \caption{\textbf{Percentage global sector-level impacts across scenarios.} Top panels report output loss (\%), and bottom panels report final demand loss (\%) across sectors under four scenarios. Note that the global sector-level impact remains unchanged for domestic/global partnership-based production increase strategy at a given parameter value, $\eta = 0.10$. 
}
    \label{fig:sectorwise_ranking_percentage}
\end{figure}

\begin{figure}[h!]
    \centering
    \includegraphics[width = 0.99\linewidth, clip, trim = 2cm 2cm 2cm 2cm]{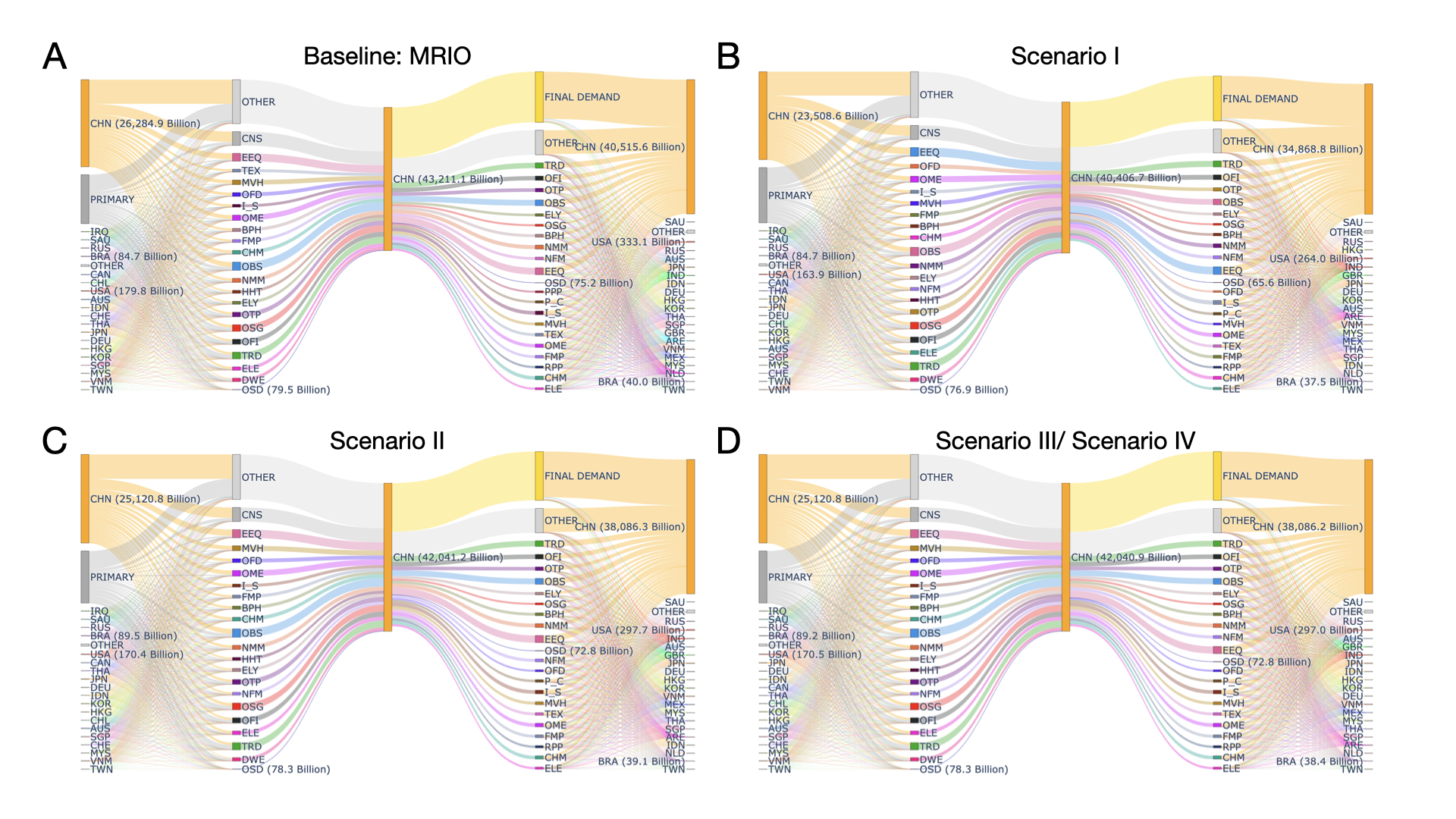}
    \caption{\textbf{Impact of disruption and mitigation scenarios on China's Supply Chain.} We illustrate the top 20 input sector and output sector flow of China, along with final demand. 
}
    \label{fig:sankey_ranking_absolute_china}
\end{figure}

\begin{figure}[h!]
    \centering
    \includegraphics[width = 0.99\linewidth, clip, trim = 2cm 2cm 2cm 2cm]{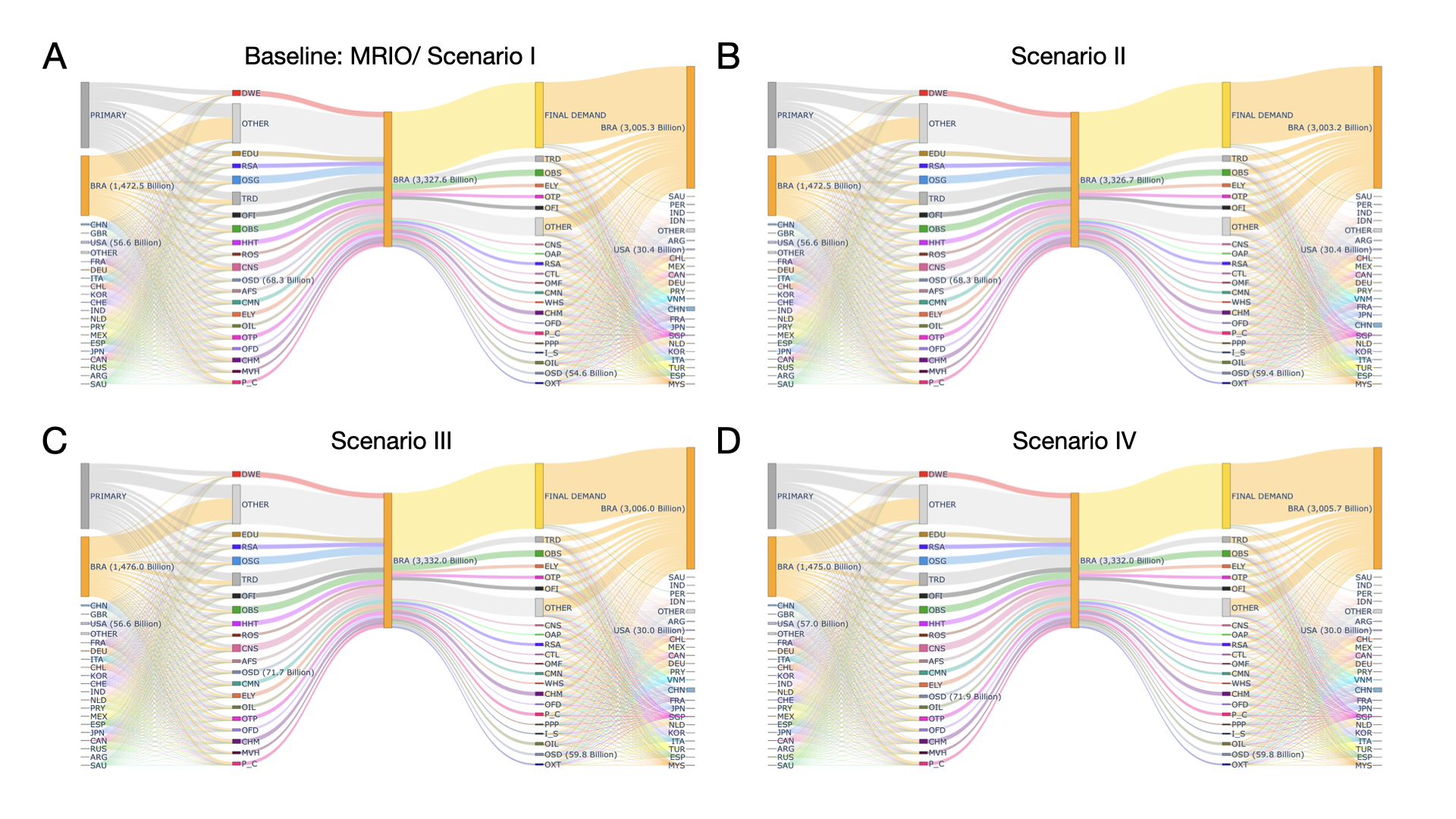}
    \caption{\textbf{Impact of disruption and mitigation scenarios on Brazil's Supply Chain.}  
}
    \label{fig:sankey_ranking_absolute_brazil}
\end{figure}

\begin{figure}[h!]
    \centering
    \includegraphics[width = 0.99\linewidth, clip, trim = 2cm 2cm 2cm 2cm]{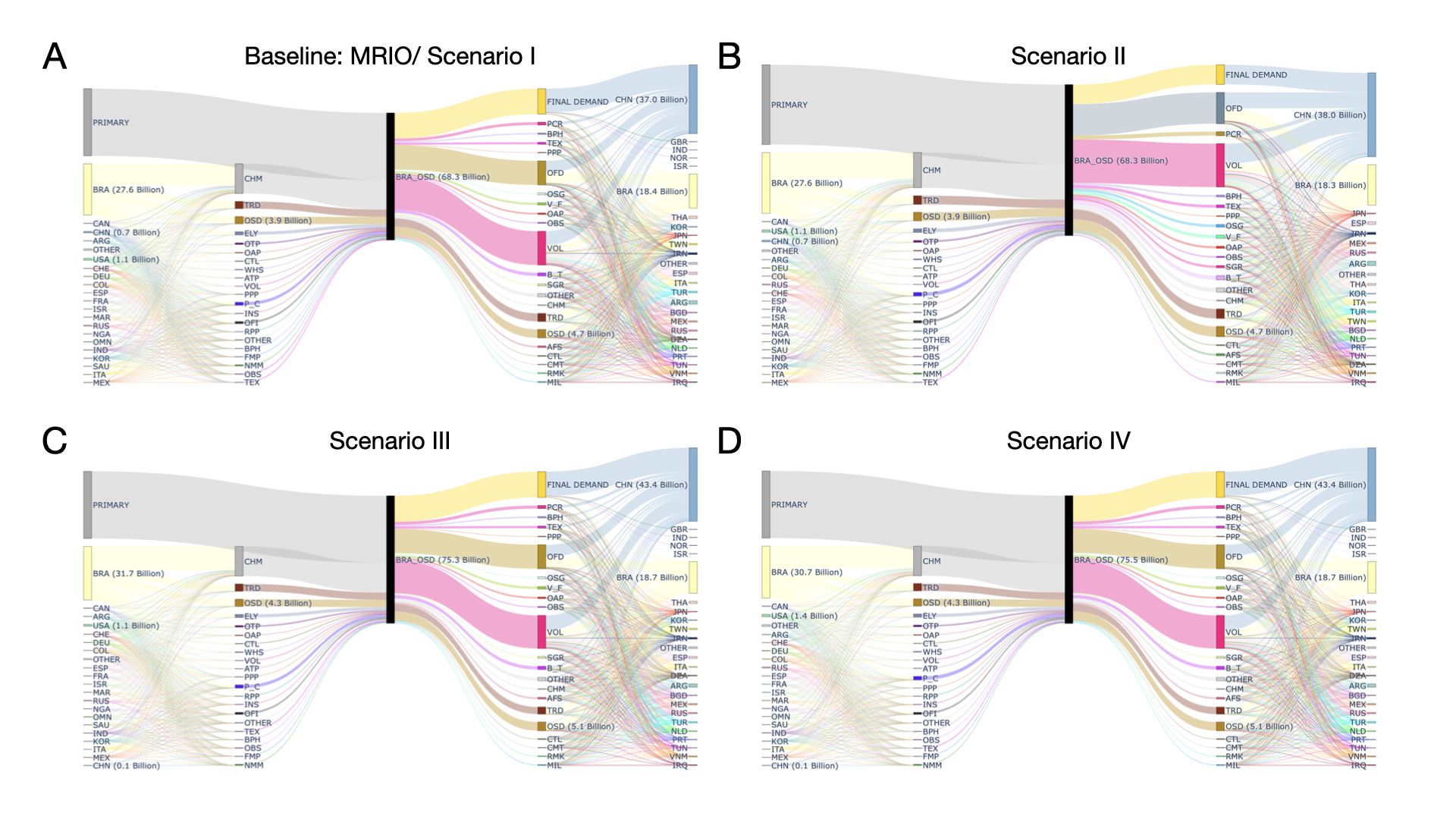}
    \caption{\textbf{Impact of disruption and mitigation scenarios on Brazil Oilseed Firm Supply Chain.}  
}
    \label{fig:sankey_ranking_absolute_brazil_osd}
\end{figure}

\begin{figure}[h!]
    \centering
    \includegraphics[width = 0.99\linewidth, clip, trim = 2cm 2cm 2cm 2cm]{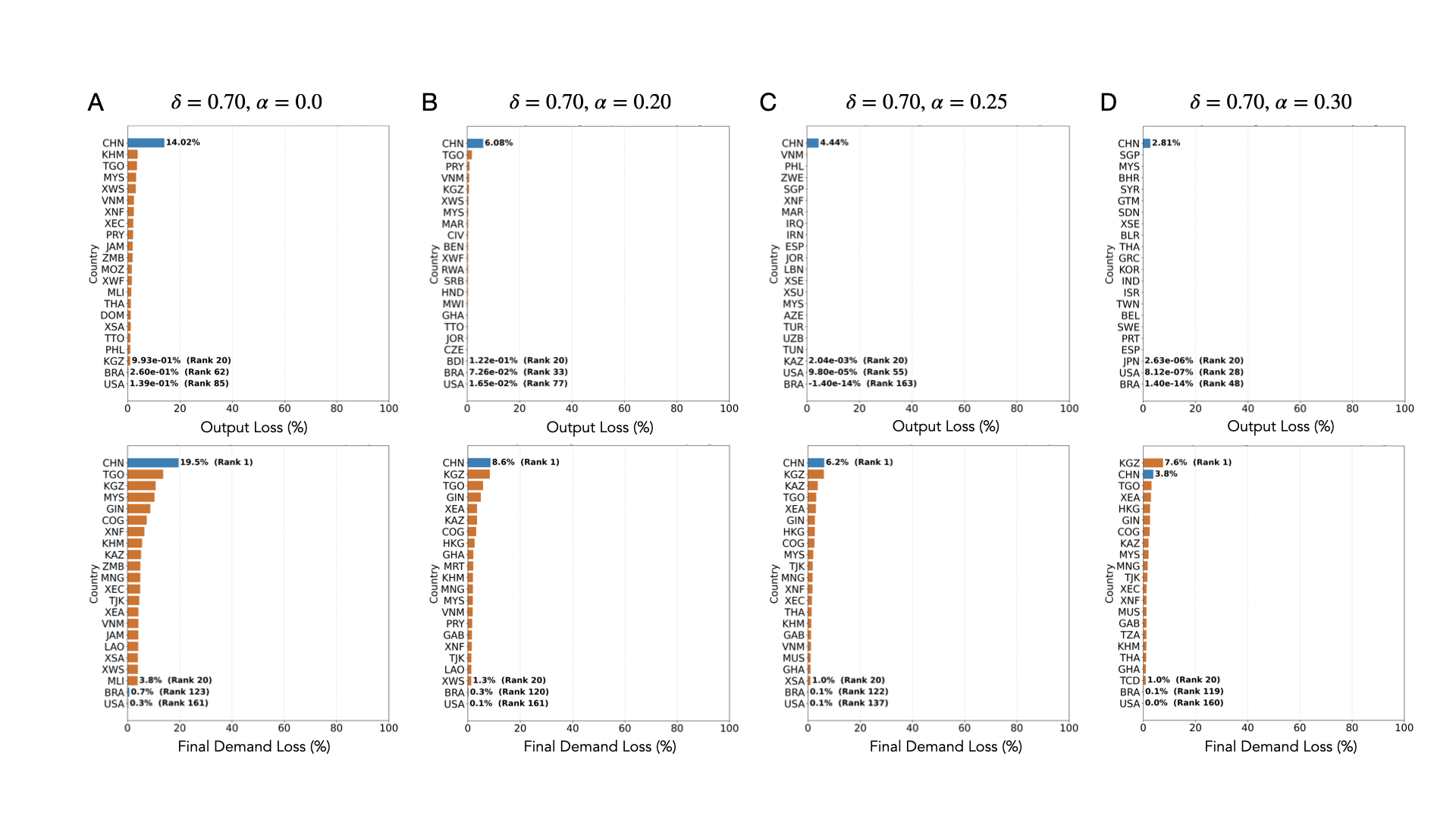}
    \caption{Percentage global country-level impacts at USA-China Oilseed flow disruption at $\delta = 0.70$ and Brazil Oilseed reallocation to China by $\alpha$ fraction.}
    \label{fig:additional_70}
\end{figure}

\begin{figure}[h!]
    \centering
    \includegraphics[width = 0.99\linewidth, clip, trim = 2cm 2cm 2cm 2cm]{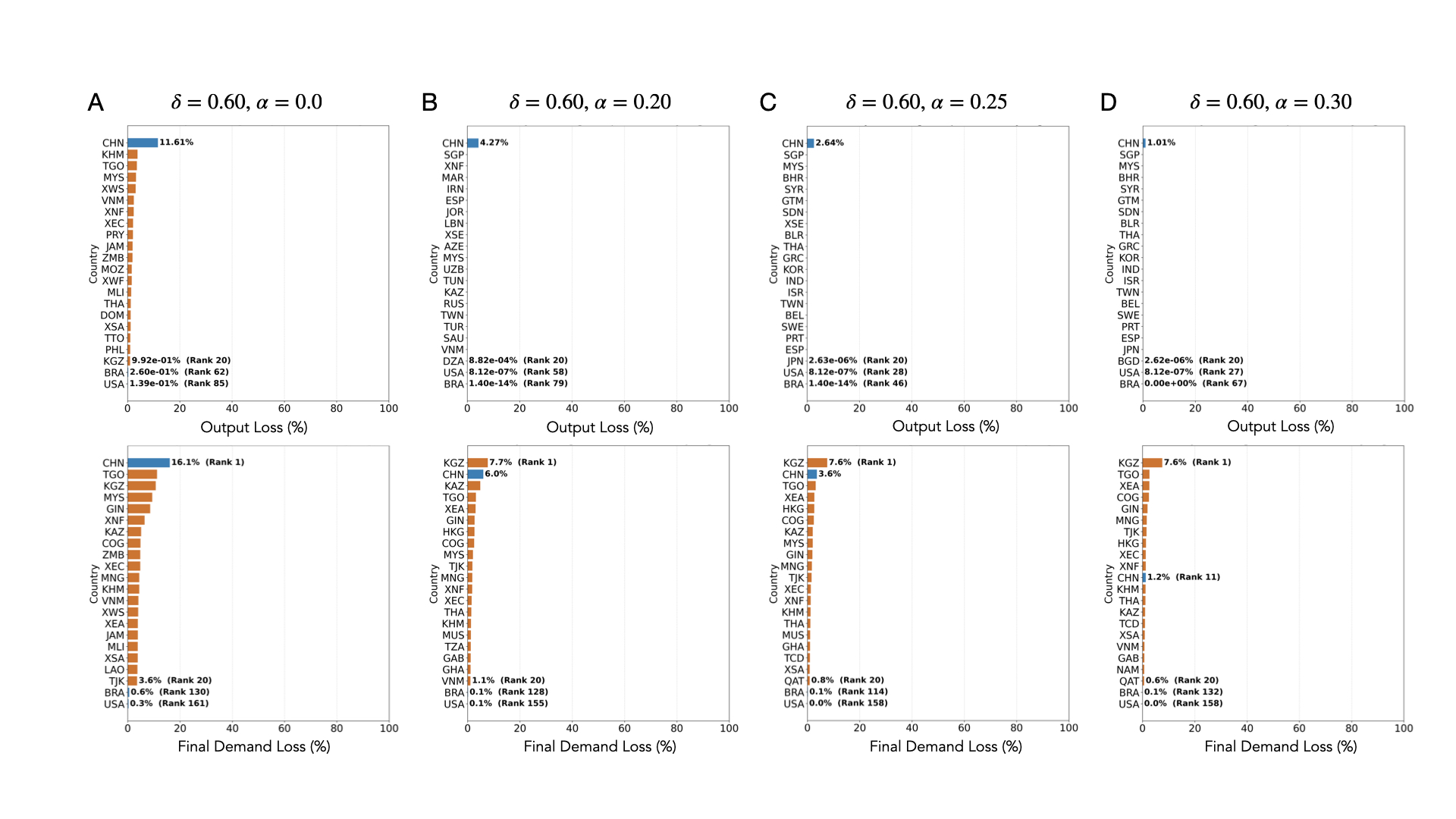}
    \caption{Percentage global country-level impacts at USA-China Oilseed flow disruption at $\delta = 0.60$ and Brazil Oilseed reallocation to China by $\alpha$ fraction.}
    \label{fig:additional_60}
\end{figure}

\newpage 
\begin{figure}[h!]
    \centering
    \includegraphics[width = 0.99\linewidth, clip, trim = 2cm 2cm 12cm 2cm]{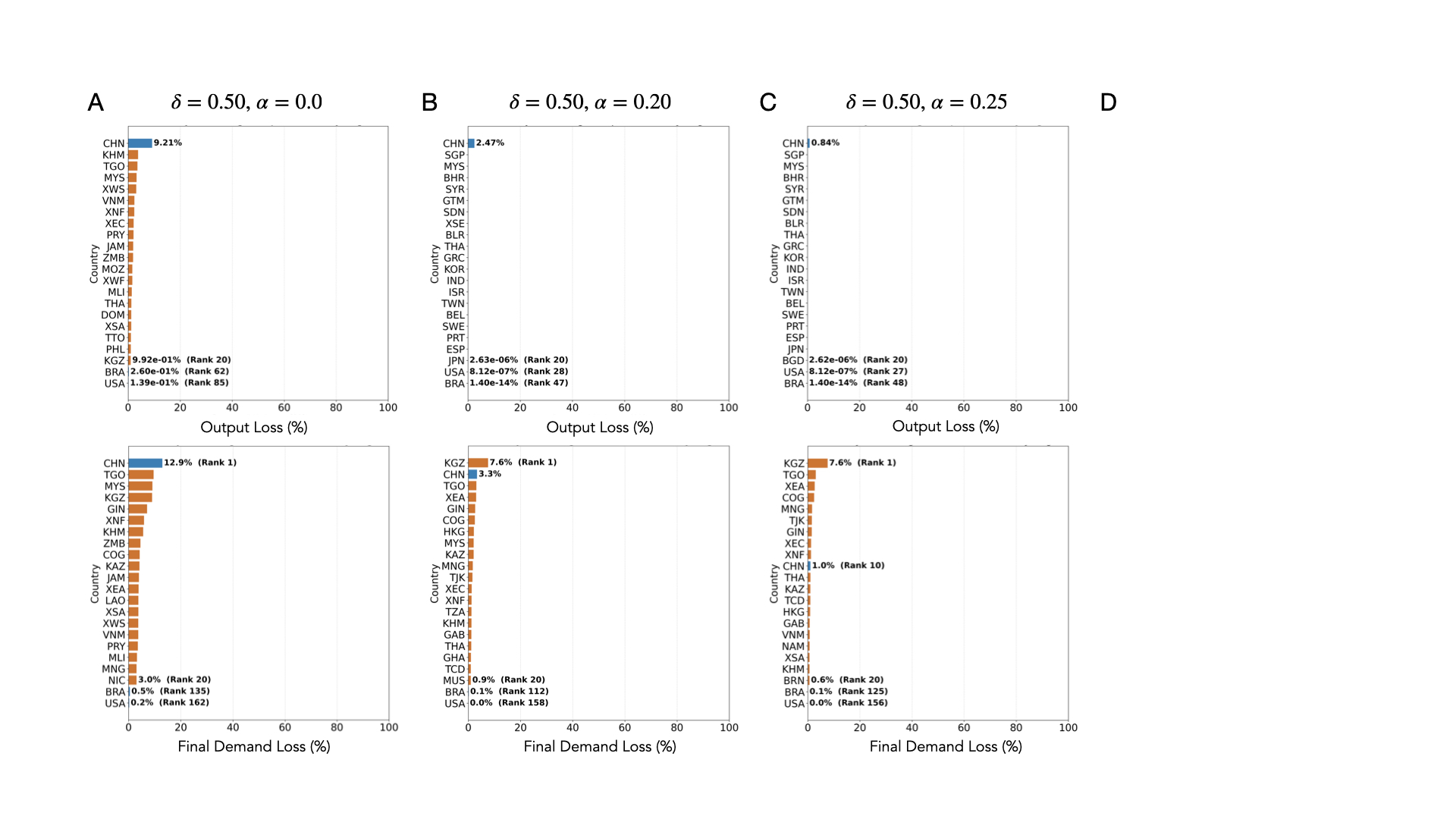}
    \caption{Percentage global country-level impacts at USA-China Oilseed flow disruption at $\delta = 0.50$ and Brazil Oilseed reallocation to China by $\alpha$ fraction.}
    \label{fig:additional_50}
\end{figure}

\begin{table*}[h]
\centering
\tiny
\setlength{\tabcolsep}{3pt}
\renewcommand{\arraystretch}{0.88}
\caption{GTAP sector codes and descriptions.}
\label{tab:GTAP_sectors}
\begin{tabularx}{\textwidth}{@{}>{\ttfamily\bfseries}l X >{\ttfamily\bfseries}l X@{}}
\toprule
\normalfont\textbf{Code} & \textbf{Description} &
\normalfont\textbf{Code} & \textbf{Description} \\
\midrule
pdr & Paddy rice & bph & Basic pharmaceutical products \\
wht & Wheat & rpp & Rubber and plastic products \\
gro & Cereal grains nec & nmm & Mineral products nec \\
v\_f & Vegetables, fruit, nuts & i\_s & Ferrous metals \\
osd & Oil seeds & nfm & Metals nec \\
c\_b & Sugar cane, sugar beet & fmp & Metal products \\
pfb & Plant-based fibers & ele & Computer, electronic and optical products \\
ocr & Crops nec & eeq & Electrical equipment \\
ctl & Cattle, sheep and goats, horses & ome & Machinery and equipment nec \\
oap & Animal products nec & mvh & Motor vehicles and parts \\
rmk & Raw milk & otn & Transport equipment nec \\
wol & Wool, silk-worm cocoons & omf & Manufactures nec \\
frs & Forestry & ely & Electricity \\
fsh & Fishing & gdt & Gas manufacture, distribution \\
coa & Coal & wtr & Water \\
oil & Oil & cns & Construction \\
gas & Gas & trd & Trade \\
oxt & Other extraction & afs & Accommodation, food services \\
cmt & Meat: cattle, sheep, goats & otp & Transport nec \\
omt & Meat products nec & wtp & Water transport \\
vol & Vegetable oils and fats & atp & Air transport \\
mil & Dairy products & whs & Warehousing \\
pcr & Processed rice & cmn & Communication \\
sgr & Sugar & ofi & Financial services nec \\
ofd & Food products nec & ins & Insurance \\
b\_t & Beverages, tobacco & rsa & Real estate activities \\
tex & Textiles & obs & Business services nec \\
wap & Wearing apparel & ros & Recreational services \\
lea & Leather products & osg & Public administration \\
lum & Wood products & edu & Education \\
ppp & Paper, publishing & hht & Health services \\
p\_c & Petroleum, coal products & dwe & Dwellings \\
chm & Chemical products &  &  \\
\bottomrule
\end{tabularx}
\end{table*}

\clearpage
\begin{table*}[h!]
\centering
\tiny
\setlength{\tabcolsep}{3pt}
\renewcommand{\arraystretch}{0.88}
\caption{GTAP region codes and descriptions.}
\label{tab:gtap-region-codes}
\begin{tabularx}{\textwidth}{@{}>{\ttfamily\bfseries}l X >{\ttfamily\bfseries}l X@{}}
\toprule
\normalfont\textbf{Code} & \textbf{Description} &
\normalfont\textbf{Code} & \textbf{Description} \\
\midrule
AUS & Australia & GBR & United Kingdom \\
NZL & New Zealand & CHE & Switzerland \\
XOC & Rest of Oceania & NOR & Norway \\
CHN & China & XEF & Rest of European Free Trade Association \\
HKG & Hong Kong, Special Administrative Region of China & ALB & Albania \\
JPN & Japan & SRB & Serbia \\
KOR & Korea, Republic of & BLR & Belarus \\
MNG & Mongolia & RUS & Russian Federation \\
TWN & Chinese Taipei & UKR & Ukraine \\
XEA & Rest of East Asia & XEE & Rest of Eastern Europe \\
BRN & Brunei Darussalam & XER & Rest of Europe \\
KHM & Cambodia & KAZ & Kazakhstan \\
IDN & Indonesia & KGZ & Kyrgyzstan \\
LAO & Lao PDR & TJK & Tajikistan \\
MYS & Malaysia & UZB & Uzbekistan \\
PHL & Philippines & XSU & Rest of Former Soviet Union \\
SGP & Singapore & ARM & Armenia \\
THA & Thailand & AZE & Azerbaijan \\
VNM & Viet Nam & GEO & Georgia \\
XSE & Rest of Southeast Asia & BHR & Bahrain \\
AFG & Afghanistan & IRN & Iran, Islamic Republic of \\
BGD & Bangladesh & IRQ & Iraq \\
IND & India & ISR & Israel \\
NPL & Nepal & JOR & Jordan \\
PAK & Pakistan & KWT & Kuwait \\
LKA & Sri Lanka & LBN & Lebanon \\
XSA & Rest of South Asia & OMN & Oman \\
CAN & Canada & PSE & Palestinian Territory, Occupied \\
USA & United States of America & QAT & Qatar \\
MEX & Mexico & SAU & Saudi Arabia \\
XNA & Rest of North America & SYR & Syrian Arab Republic \\
ARG & Argentina & TUR & Turkey \\
BOL & Bolivia & ARE & United Arab Emirates \\
BRA & Brazil & XWS & Rest of Western Asia \\
CHL & Chile & DZA & Algeria \\
COL & Colombia & EGY & Egypt \\
ECU & Ecuador & MAR & Morocco \\
PRY & Paraguay & TUN & Tunisia \\
PER & Peru & XNF & Rest of North Africa \\
URY & Uruguay & BEN & Benin \\
VEN & Venezuela (Bolivarian Republic of) & BFA & Burkina Faso \\
XSM & Rest of South America & CMR & Cameroon \\
CRI & Costa Rica & CIV & Côte d'Ivoire \\
GTM & Guatemala & GHA & Ghana \\
HND & Honduras & GIN & Guinea \\
NIC & Nicaragua & MLI & Mali \\
PAN & Panama & MRT & Mauritania \\
SLV & El Salvador & NER & Niger \\
XCA & Rest of Central America & NGA & Nigeria \\
DOM & Dominican Republic & SEN & Senegal \\
HTI & Haiti & TGO & Togo \\
JAM & Jamaica & XWF & Rest of Western Africa \\
PRI & Puerto Rico & AGO & Angola \\
TTO & Trinidad and Tobago & CAF & Central African Republic \\
XCB & Rest of Caribbean & TCD & Chad \\
AUT & Austria & COG & Congo \\
BEL & Belgium & COD & Democratic Republic of the Congo \\
BGR & Bulgaria & GNQ & Equatorial Guinea \\
HRV & Croatia & GAB & Gabon \\
CYP & Cyprus & STP & Sao Tome and Principe \\
CZE & Czech Republic & BDI & Burundi \\
DNK & Denmark & COM & Comoros \\
EST & Estonia & ETH & Ethiopia \\
FIN & Finland & KEN & Kenya \\
FRA & France & MDG & Madagascar \\
DEU & Germany & MWI & Malawi \\
GRC & Greece & MUS & Mauritius \\
HUN & Hungary & MOZ & Mozambique \\
IRL & Ireland & RWA & Rwanda \\
ITA & Italy & SDN & Sudan \\
LVA & Latvia & TZA & Tanzania, United Republic of \\
LTU & Lithuania & UGA & Uganda \\
LUX & Luxembourg & ZMB & Zambia \\
MLT & Malta & ZWE & Zimbabwe \\
NLD & Netherlands & XEC & Rest of Eastern Africa \\
POL & Poland & BWA & Botswana \\
PRT & Portugal & SWZ & Eswatini \\
ROU & Romania & NAM & Namibia \\
SVK & Slovakia & ZAF & South Africa \\
SVN & Slovenia & XSC & Rest of South African Customs Union \\
ESP & Spain & XTW & Rest of the World \\
SWE & Sweden &  &  \\
\bottomrule
\end{tabularx}
\end{table*}

\clearpage

\end{document}